\documentclass[subeqn,11pt, a4paper]{amsart}

\usepackage[margin=27mm]{geometry}

\usepackage{amssymb,amsmath,amsfonts,amsthm, amscd, mathrsfs, helvet, mathtools}
\usepackage{bm, stmaryrd}
\usepackage{graphicx, verbatim}
\usepackage[mathcal]{euscript}
\usepackage{color}
\usepackage[all]{xy}
\usepackage{relsize}

\usepackage{tikz}
\usetikzlibrary{cd}
\usetikzlibrary{matrix,calc}
\usetikzlibrary{decorations.markings,shapes.geometric,shapes.misc}
\usetikzlibrary{decorations.pathreplacing,positioning}
\usetikzlibrary{arrows}

\tikzset{cross/.style={cross out, draw=black, minimum size=2*(#1-\pgflinewidth), inner sep=0pt, outer sep=0pt}, cross/.default={1pt}}

\tikzset{cross/.style={cross out, draw=black, minimum size=2*(#1-\pgflinewidth), inner sep=0pt, outer sep=0pt}, cross/.default={1pt}}

\allowdisplaybreaks[1]

\usepackage{pict2e}

\usepackage{hyperref}
\hypersetup{
     colorlinks=false,         
     linkcolor=darkred,
     citecolor=blue,
}



\DeclareRobustCommand{\SkipTocEntry}[5]{}

\definecolor{myPurple}{rgb}{0.5,0.1,0.6}
\definecolor{myOrange}{rgb}{1.0,0.5,0.0}
\definecolor{myRed}{rgb}{1.0,0.0,0.0}
\definecolor{myGreen}{rgb}{0.0,0.5,0.0}
\definecolor{LatexBlue}{rgb}{0.211765,0.227451,0.666667}
\definecolor{myBlue}{rgb}{0.0,0.0,1.0}
\definecolor{myBlack}{rgb}{0.0,0.0,0.0}
\definecolor{myGray}{rgb}{0.3,0.3,0.3}

\theoremstyle{plain}
\newtheorem{theorem}{Theorem}[section]
\newtheorem*{theorem*}{Theorem}
\newtheorem{proposition}[theorem]{Proposition}
\newtheorem*{proposition*}{Proposition}
\newtheorem{lemma}[theorem]{Lemma}

\theoremstyle{definition}
\newtheorem{definition}[theorem]{Definition}
\newtheorem{assumption}[theorem]{Assumption}

\newtheorem{remarkx}[theorem]{Remark}

\newenvironment{remark}
  {\pushQED{\qed}\remarkx}
  {\popQED\endremarkx}

\newcommand{\tensor}[1]{{\mathfrak{#1}}}

\DeclareMathOperator{\Aut}{Aut}

\DeclareMathOperator{\Ad}{Ad}
\DeclareMathOperator{\ad}{ad}

\newcommand{\longhookrightarrow}{\lhook\joinrel\relbar\joinrel\rightarrow}

\def\dd{{\rm d}}

\def\g{\mathfrak{g}}

\def\id{\textup{id}}
\def\h{\mathfrak{h}}

\def\CS{\mathrm{CS}}

\def\bbR{\mathbb{R}}
\def\bbC{\mathbb{C}}

\newcommand{\lau}[1]{(\kern-.2em( #1 )\kern-.2em)}

\newcommand{\ip}[2]{\left\langle #1 , #2 \right\rangle}
\newcommand{\ipp}[2]{\left\langle\!\left\langle #1 , #2 \right\rangle\!\right\rangle}

\def\A{\mathcal{A}}

\def\PN{\mathbb{PN}}
\def\PT{\mathbb{PT}}
\def\CP{\mathbb{C}P^1}

\def\RR{\mathbb{R}}

\def\ZZ{\mathbb{Z}}

\def\ii{{\rm i}}

\def\1{\tensor{1}}
\def\2{\tensor{2}}
\def\3{\tensor{3}}
\def\4{\tensor{4}}

\numberwithin{equation}{section}

\def\Xint#1{\mathchoice
{\XXint\displaystyle\textstyle{#1}}%
{\XXint\textstyle\scriptstyle{#1}}%
{\XXint\scriptstyle\scriptscriptstyle{#1}}%
{\XXint\scriptscriptstyle\scriptscriptstyle{#1}}%
\!\int}
\def\XXint#1#2#3{{\setbox0=\hbox{$#1{#2#3}{\int}$ }
\vcenter{\hbox{$#2#3$ }}\kern-.6\wd0}}

\def\dashint{\Xint-}

\begin{document}

\title[5d 2-CS and 3d IFTs]{5d 2-Chern-Simons theory and\\
3d integrable field theories}

\author{Alexander Schenkel}
\address{School of Mathematical Sciences, University of Nottingham, University Park, Nottingham NG7 2RD, United Kingdom.}
\email{alexander.schenkel@nottingham.ac.uk}
\author{Beno\^{\i}t Vicedo}
\address{Department of Mathematics, University of York, York YO10 5DD, United Kingdom.}
\email{benoit.vicedo@gmail.com}

\begin{abstract}
The $4$-dimensional semi-holomorphic Chern-Simons theory of Costello and Yamazaki provides a gauge-theoretic origin for the Lax connection of $2$-dimensional integrable field theories. The purpose of this paper is to extend this framework to the setting of $3$-dimensional integrable field theories by considering a $5$-dimensional semi-holomorphic higher Chern-Simons theory for a higher connection $(A,B)$ on $\mathbb{R}^3 \times \mathbb{C}P^1$. The input data for this theory are the choice of a meromorphic $1$-form $\omega$ on $\mathbb{C}P^1$ and a strict Lie $2$-group with cyclic structure on its underlying Lie $2$-algebra. Integrable field theories on $\mathbb{R}^3$ are constructed by imposing suitable boundary conditions on the connection $(A,B)$ at the $3$-dimensional defects located at the poles of $\omega$ and choosing certain admissible meromorphic solutions of the bulk equations of motion. The latter provides a natural notion of higher Lax connection for $3$-dimensional integrable field theories, including a $2$-form component $B$ which can be integrated over Cauchy surfaces to produce conserved charges. As a first application of this approach, we show how to construct a generalization of Ward's $(2+1)$-dimensional integrable chiral model from a suitable choice of data in the $5$-dimensional theory.
\end{abstract}

\maketitle

\setcounter{tocdepth}{1}
\tableofcontents

\input{epsf}

\section{Introduction}
Although there is no universally agreed definition of integrability in the
context of field theories, a hallmark of integrability in this infinite dimensional
setting is the existence of infinitely many independent conserved charges.
One very general and powerful framework allowing for the systematic construction
of these conserved charges is the Lax formalism \cite{Lax}, in which the field equations
are expressed as the consistency condition for an overdetermined system of
linear partial differential equations. An important
incarnation of this formalism is in the context of $2$-dimensional field
theories where a \emph{Lax connection} is defined as an on-shell flat connection that
depends meromorphically on an auxiliary Riemann surface $C$, which is usually
taken to be the Riemann sphere $\CP$. If a $2$-dimensional field theory admits a
Lax connection, then its holonomy along curves of constant time, which depends
analytically on $C$, serves as a generating function for infinitely many conserved charges.
Unfortunately, a suitable Lax connection for a given $2$-dimensional field theory is
typically found by clever guesswork, making its origin quite obscure.

\medskip

In their seminal paper \cite{Costello:2019tri}, Costello and Yamazaki 
gave a very elegant gauge-theoretic origin for the Lax connection in 
$2$-dimensional integrable field theories which is based on a $4$-dimensional
semi-holomorphic variant of Chern-Simons theory 
\cite{Nekrasov,Costello:2013zra,Costello:2013sla, Witten:2016spx, Costello:2017dso, Costello:2018gyb}.
In this approach, integrable field theories on a $2$-dimensional
manifold $\Sigma$, with Lax connection depending meromorphically on a Riemann surface $C$,
arise as specific solutions to this $4$-dimensional semi-holomorphic Chern-Simons
theory on $\Sigma \times C$.
The Lagrangian of the latter is given by $\omega \wedge \CS(A)$,
where $\omega$ is a fixed meromorphic $1$-form on $C$ and $\CS(A)$ is the
Chern-Simons $3$-form for a $\g$-valued $1$-form $A$ on $\Sigma \times C$.
A $2$-dimensional integrable field theory is then determined by the choice of $1$-form $\omega$
and of boundary conditions imposed on the gauge field $A$ at the surface defects
$\Sigma \times \{ x \} \subset \Sigma \times C$ located at each pole $x$ of the $1$-form $\omega$.
Importantly, the Lax connection emerges naturally as a meromorphic solution of the bulk equations of motion
for the gauge field $A$, with pole structure determined by the zeros of $\omega$ and
satisfying the chosen boundary conditions.

This relatively recent gauge-theoretic approach to $2$-dimensional integrable field theories has
already established itself as a very powerful tool for constructing new $2$-dimensional classical
integrable field theories, leading to the discovery of vast new families of examples.
This framework also serves as a very efficient organizational
tool for navigating the ever expanding zoo of $2$-dimensional integrable field theories.
See for instance \cite{Delduc:2019whp, Schmidtt:2019otc, Bassi:2019aaf, Fukushima:2020kta, Hoare:2020mpv, Lacroix:2020flf, Caudrelier:2020xtn, Fukushima:2020tqv, He:2021xoo, Fukushima:2021eni, Fukushima:2021ako, Liniado:2023uoo, Berkovits:2024reg, Lacroix:2024wrd}.

\medskip

In stark contrast to this extremely rich $2$-dimensional setting, there are currently very few known
examples of integrable field theories in higher dimensions. Perhaps the most well-known examples
are the Kadomtsev–Petviashvili (KP) equation in $3$ dimensions and the anti-self-dual Yang-Mills (ASDYM)
equation in $4$ dimensions. Another less well-known example, which will be particularly relevant for us later,
is Ward's equation \cite{Wmodel1, Wmodel2} which describes a non-relativistic modification of the non-integrable
$3$-dimensional chiral model that was obtained as a reduction of the ASDYM equation.
The KP equation admits a more exotic Lax formalism which roughly speaking encodes the
third dimension using the language of pseudo-differential operators, see for instance \cite{BBT}
for an extensive review. The Lax formalism used to encode the ASDYM equation, on the other hand, is based
on a partial flatness condition for a certain connection on $\mathbb{R}^4$. The latter lies at the
heart of the Penrose-Ward correspondence \cite{PWcorr} which relates solutions of the ASDYM
equation on $\RR^4$ to certain holomorphic vector bundles on twistor space
$\PT = \mathbb{C}P^3 \setminus \CP$.

Inspired by the Penrose-Ward correspondence, and following a proposal of Costello,
it was shown by Bittleston and Skinner in \cite{Bittleston:2020hfv} (see also the related works
\cite{Penna:2020uky, Cole:2023umd}) that various known Lagrangians for the ASDYM equation can be derived
starting from a $6$-dimensional holomorphic Chern-Simons theory on twistor space $\PT$. More specifically,
the Lagrangian for the latter is given by $\Omega \wedge \mathsf{hCS}(\mathcal A)$, where $\Omega$
is a fixed meromorphic $(3,0)$-form on $\PT$ and $\mathsf{hCS}(\mathcal A)$ is the holomorphic
Chern-Simons $(0,3)$-form for a $\g$-valued $(0,1)$-form $\mathcal A$ on $\PT$. The inevitable presence of
poles in $\Omega$ leads to a violation of gauge invariance which can nevertheless be restored by imposing suitable
boundary conditions on $\mathcal A$ at these poles. The ASDYM equation then emerges from
the simplest choice of $\Omega$, with a pair of double poles and without zeros, in much the same
way as $2$-dimensional integrable field theories emerge from $4$-dimensional semi-holomorphic Chern-Simons theory.
Interestingly, it was also shown in \cite{Bittleston:2020hfv} that applying the same procedure starting from
more general $(3,0)$-forms $\Omega$ leads to other $4$-dimensional actions with field
equations generalizing the ASDYM equation and which are again manifestly integrable by
the Penrose-Ward correspondence.

\medskip

It is therefore tempting to regard $6$-dimensional holomorphic Chern-Simons theory on $\PT$
as playing an analogous role to $4$-dimensional semi-holomorphic Chern-Simons theory for describing
$4$-dimensional integrable field theories.
However, although the ASDYM equation is certainly a $4$-dimensional integrable field theory \--- it is exactly
solvable by the ADHM construction \cite{ADHM} which is rooted in the Penrose-Ward
correspondence \--- its Lax formalism lacks certain features one would expect of a 
higher-dimensional integrable field theory.

Recall that the defining properties of the Lax connection for a $2$-dimensional integrable field theory
ensure that its holonomy along curves of constant time is both conserved and depends analytically on an
auxiliary Riemann surface $C$.
For a field theory on a $(d+1)$-dimensional spacetime with $d \geq 1$, a conserved charge should be
given by the integral of a $d$-form over a $d$-dimensional submanifold representing a constant time slice.
One should therefore expect an adequate notion of Lax connection for a $(d+1)$-dimensional integrable field theory
to involve a $d$-form. This observation was already made nearly 30 years ago in \cite{Alvarez:1997ma}
and further explored in subsequent works, see for instance \cite{Gianzo:1998ez, Adam:2008jx}.
However, in these works, although the right notion of higher connections was used, a meromorphic
dependence on an auxiliary Riemann surface was never considered. For a recent review of existing
descriptions of classical integrable field theories in $3$ dimensions see \cite{Gubarev:2023jtp}.

\medskip

The goal of this paper is to initiate the exploration of integrable field theories in higher dimensions
using higher gauge-theoretic methods. Specifically, we will focus in this paper on the problem of constructing
$3$-dimensional integrable field theories. Our proposal is to start from a $5$-dimensional higher gauge
theory variant of the $4$-dimensional semi-holomorphic Chern-Simons theory from \cite{Costello:2019tri}
which is defined on a product manifold $X=M\times C$ with $M$ representing a $3$-dimensional
spacetime and $C$ a Riemann surface. The structure group $G$ from the ordinary approach 
is generalized to a strict Lie $2$-group which we describe explicitly in terms of
a crossed module of Lie groups $(G,H,t,\alpha)$, see e.g.\ 
\cite{SchreiberWaldorf,Joao,Waldorf1,Waldorf2,Saemann} and also Section \ref{sec:prelims}
for a review. A connection (or gauge field) 
in this context is given by a pair $(A,B)\in \Omega^1(X,\g)\times \Omega^2(X,\h)$
consisting of both a $1$-form and a $2$-form taking values in the underlying Lie $2$-algebra
of the structure Lie $2$-group. Such higher connections have $1$-dimensional and also 
$2$-dimensional holonomies (see e.g.\ \cite{SchreiberWaldorf,Joao,Waldorf2}), 
which provides additional flexibility for the construction of conserved charges in a 
$3$-dimensional integrable field theory, see also Subsection \ref{subsec:2hol}. The Lagrangian of our theory takes
the form $\omega\wedge \CS(A,B)$, where $\omega$ is a fixed meromorphic $1$-form
on $C$ and $\CS(A,B)$ is the $2$-Chern-Simons $4$-form for the higher connection $(A,B)$
which is associated with the choice of a suitable non-degenerate invariant pairing on the Lie $2$-algebra,
see e.g.\ \cite[Section 5.2]{Jurco} and also Section \ref{sec:prelims} for a review.
In analogy to the case of $4$-dimensional semi-holomorphic Chern-Simons theory,
the action of our theory is not automatically gauge-invariant so that one has to 
impose suitable boundary conditions for the connection $(A,B)$ on the $3$-dimensional defects
$M\times\{x\}\subset X$ located at each pole $x$ of $\omega$. Meromorphic solutions
to the Euler-Lagrange equations of this action naturally 
provide flat higher connections $(A,B)$ on $M$ which depend meromorphically
on $C$, i.e.\ candidates for Lax connections for $3$-dimensional integrable field theories on $M$.
In our approach, an integrable field theory is specified by the choice of 
1.)~a structure Lie $2$-group with
non-degenerate invariant pairing on its Lie $2$-algebra, 
2.)~a meromorphic $1$-form $\omega$ on $C$, and 3.)~suitable
boundary conditions for $(A,B)$ at the defects located at each pole $x$ of $\omega$.

\medskip

We will now describe in more detail our results by outlining
the content of this paper. In Section \ref{sec:prelims},
we provide a brief introduction to higher gauge theory.
The reader can find more details in the 
articles \cite{SchreiberWaldorf,Joao,Waldorf1,Waldorf2} and the review \cite{Saemann}.
This includes a quick recap of crossed modules of Lie groups and Lie algebras 
(Subsection \ref{subsec:crossedmod}), higher gauge fields and their gauge transformations
(Subsection \ref{subsec:2grpd}), higher parallel transports (Subsection \ref{subsec:2hol}) 
and the construction of the $2$-Chern-Simons $4$-form
(Subsection \ref{subsec:2CSform}). The material presented in this section is rather
standard and well-known, probably with the exception of the gauge transformation
property of the $2$-Chern-Simons $4$-form in Proposition \ref{propo:2CS 4form gauge},
which was derived earlier in \cite{Zucchini:2021bnn} only under additional assumptions on the 
crossed module. We would like to emphasize that it is important for us to consider
also connections $(A,B)$ which \textit{do not} necessarily satisfy the so-called
fake-flatness condition $\dd A+\tfrac{1}{2}[A,A] - t_\ast(B)=0$
since otherwise our action functional would degenerate. This pushes
us out of the standard framework for higher connections developed in 
\cite{SchreiberWaldorf,Joao,Waldorf1,Waldorf2}, which as a 
consequence prevents us from considering also $2$-gauge transformations
between gauge transformations, see Remark \ref{rem:2gauge}. (More concisely,
this means that our non-fake-flat higher connections only form a groupoid and not a $2$-groupoid.)
There are recent developments towards a theory of non-fake-flat higher connections
through so-called adjusted connections, see e.g.\ 
\cite{AdjustmentsRSW,Kim:2019owc,Tellez-Dominguez}, 
but the additional adjustment data seem to be incompatible 
with the type of boundary conditions we would like to impose on our connections,
see Subsections \ref{subsec:bdy simple} and \ref{subsec:edge modes simple}.
These inconveniences associated with non-fake-flat connections disappear once we go
on-shell since solutions of our $5$-dimensional semi-holomorphic $2$-Chern-Simons theory
are fully flat connections, and hence in particular fake-flat. This means that 
the construction of $2$-dimensional holonomies from \cite{SchreiberWaldorf,Joao,Waldorf1,Waldorf2},
which are needed in our approach to generate conserved charges, is directly applicable
in our context once we go on-shell.

In Section \ref{sec:simplepoles}, we study in detail our $5$-dimensional semi-holomorphic
$2$-Chern-Simons theory in the special case where the meromorphic $1$-form
$\omega$ has only simple poles. In Subsection \ref{subsec:action simple pole}
we spell out concretely the action functional for this theory
and in Subsection \ref{subsec:gauge simple} we analyze the properties
of this action under gauge transformations. It is shown that there
are violations to gauge invariance which are localized at the $3$-dimensional 
defects $M\times\{x\}\subset X$ located at the poles $x$ of $\omega$, see 
Proposition \ref{propo:bdy gauge violation}. In Subsection \ref{subsec:bdy simple}
we restore gauge invariance by imposing suitable boundary conditions
which are determined by the choice of an isotropic crossed submodule 
$(G^\diamond,H^\diamond,t^{\bm{z}},\alpha^{\bm{z}})
\subseteq (G^{\bm{z}},H^{\bm{z}},t^{\bm{z}},\alpha^{\bm{z}})$ of 
the crossed module associated with the defect. We will show in
Subsection \ref{subsec:edge modes simple} that these boundary conditions
admit an equivalent (homotopical) interpretation in terms of edge mode
fields living on the defect, which take the form of pairs 
$(k,\kappa)$ with $k\in C^\infty(M,G^{\bm{z}})$ a group-valued function
and $\kappa\in\Omega^1(M,\h^{\bm{z}})$ a Lie algebra valued $1$-form
on the $3$-dimensional spacetime $M$. While the group-valued
edge modes $k$ are familiar from $4$-dimensional semi-holomorphic Chern-Simons theory, see 
e.g.\ \cite{Benini:2020skc}, the $1$-form edge modes $\kappa$ are a novel feature 
of our higher gauge theoretic approach to $3$-dimensional integrable field theories.
Our approach results in an explicit action functional \eqref{eqn:extended action simple} 
for the edge modes on $M$ which gives them their dynamics.
In Subsection \ref{subsec:EOM} we derive the Euler-Lagrange equations
of our extended bulk+defect action functional \eqref{eqn:extended action simple}
and find solutions which describe flat connections on $M$ that are meromorphic on $C$,
as required for a Lax connection.

In Section \ref{sec:higherpoles}, we generalize the results from Section \ref{sec:simplepoles}
to the case where $\omega$ has poles of arbitrary order. Using the concept of 
regularized integrals from \cite{Li:2020ljm} and \cite{Benini:2020skc},
this is easily achievable and does not pose any additional challenges. 
The main new feature of the higher-order pole case is that the defect crossed
module $(G^{\hat{\bm z}},H^{\hat{\bm z}},t^{\hat{\bm z}},\alpha^{\hat{\bm z}})$ consists
of products of jet groups associated with $(G,H,t,\alpha)$.

The aim of Section \ref{sec:integrable} is to apply our approach to 
construct explicit examples of $3$-dimensional integrable field theories.
The key concept which enables these constructions
is that of admissible connections from \cite{Benini:2020skc},
which we generalize to our present context of higher gauge theory.
In Subsection \ref{subsec:degreecounting} we introduce a suitable concept
of maximality for isotropic crossed submodules under which it can be 
expected that the Lax connection $(A,B)$
can be expressed uniquely in terms of the edge mode fields $(k,\kappa)$.
In Subsection \ref{subsec:ChernSimons} we present a toy-model
to illustrate our proposed construction of $3$-dimensional integrable
field theories by focusing on one of the simplest choices of $\omega$ 
given by a meromorphic $1$-form with a single zero, a simple pole and a double pole.
The resulting $3$-dimensional field theory in this case is given by
Chern-Simons theory. This is integrable in the sense that the Chern-Simons
equations of motion, i.e.\ flatness of the connection, arise as the flatness
of the associated Lax connection, but the Lax connection in this toy-example
is rather trivial from the perspective integrable field theory
since it is constant in the complex coordinate $z\in C$.

A more interesting example is presented in Subsection \ref{subsec:Ward}, where the
meromorphic $1$-form $\omega$ is taken to have four simple zeros and three double poles.
By choosing a suitable isotropic crossed submodule, we derive equations of motion
for the edge mode fields and observe that they are related to Ward's equation, also
known as the integrable chiral model \cite{Wmodel1, Wmodel2}.
In particular, our approach provides a direct derivation of
Ward's consistency assumption that the distinguished vector in Ward's equation
is normalized and spacelike. The Ward equation was
originally obtained from a particular choice of gauge in the Yang-Mills-Higgs system on $\mathbb{R}^{2,1}$,
which itself arises as a symmetry reduction of the ASDYM equation on $\mathbb{R}^{2,2}$
by the action of a one-parameter group of non-null translations. It was shown in \cite{Bittleston:2020hfv}
that the same symmetry reduction applied to $6$-dimensional holomorphic Chern-Simons theory
on twistor space $\PT$ leads to a partially holomorphic $5$-dimensional variant of Chern-Simons theory
on the quotient $\PN$ of $\PT$ by this group of translations. Moreover, it was
mentioned in \cite{Bittleston:2020hfv} that Ward's model can be naturally obtained
from this $5$-dimensional partially holomorphic Chern-Simons theory. We stress,
however, that there is no immediate relationship
between our $5$-dimensional $2$-Chern-Simons theory, which is based on higher gauge fields
$A \in \Omega^1(X, \g)$ and $B \in \Omega^2(X, \h)$, and the
$5$-dimensional theory considered in \cite{Bittleston:2020hfv}, which is based on
an ordinary gauge field $A' \in \Omega^1(\PN, \g)$. In fact, the model
we construct in Subsection \ref{subsec:Ward} is a considerable generalization
of the Ward model which involves three $1$-form fields valued in $\h$ as well as
the $G$-valued field of the original Ward model. We expect that the presence of
$\h$-valued $1$-form fields is a general feature of $3$-dimensional integrable field theories
constructed via our approach, which is ultimately connected to the existence
of a $2$-form component $B \in \Omega^2(X, \h)$ of the higher Lax connection.

\subsubsection*{Acknowledgments}
We would like to thank Severin Bunk, Christian S\"amann and Konrad Waldorf
for useful discussions about higher connections.
A.S.\ gratefully acknowledges the support of 
the Royal Society (UK) through a Royal Society University 
Research Fellowship (URF\textbackslash R\textbackslash 211015)
and Enhancement Grants (RF\textbackslash ERE\textbackslash 210053 and 
RF\textbackslash ERE\textbackslash 231077). B.V.\ gratefully acknowledges
the support of the Leverhulme Trust 
through a Leverhulme Research Project Grant (RPG-2021-154).


\section{\label{sec:prelims}Preliminaries on higher gauge theory}
In this section we recall some basic aspects of higher gauge theory, 
i.e.\ the theory of connections on higher-categorical analogues of principal bundles
in which the structure group is generalized from a Lie group to a Lie $2$-group.
We consider only strict Lie $2$-groups, which we describe in terms of crossed modules of Lie groups, 
and globally trivial principal $2$-bundles. More details about higher gauge theory can be found in
the articles \cite{SchreiberWaldorf,Joao,Waldorf1,Waldorf2} and the review \cite{Saemann}.

\subsection{\label{subsec:crossedmod}Crossed modules of Lie groups and Lie algebras}
A convenient and computationally efficient model for strict Lie $2$-groups 
is given by the following
\begin{definition}
A \textit{crossed module of Lie groups} 
is a tuple $(G, H, t, \alpha)$ consisting of two Lie groups $G$ and $H$, 
a Lie group homomorphism $t : H \to G$, 
and a smooth action $\alpha : G \times H \to H$ of $G$ on $H$ 
in terms of Lie group automorphisms,
such that
\begin{subequations} \label{eqn:CM axioms Lie}
\begin{flalign}
\label{eqn:CM axiom Lie a} t\big( \alpha(g, h) \big) \,&=\, g \,t(h)\, g^{-1}\quad, \\
\label{eqn:CM axiom Lie b} \alpha\big( t(h), h^\prime \big) \,&=\, h \,h^\prime\, h^{-1}\quad,
\end{flalign}
\end{subequations}
for all $g \in G$ and $h, h^\prime \in H$.
\end{definition}

Associated to each crossed module of Lie groups $(G,H,t,\alpha)$
is a \textit{crossed module of Lie algebras} $(\g,\h,t_\ast,\alpha_\ast)$
which models the Lie $2$-algebra of the corresponding strict Lie $2$-group.
Here $\g$ and $\h$ denote the Lie algebras of, respectively, $G$ and $H$.
The Lie algebra homomorphism 
\begin{flalign} \label{eqn:t star}
t_\ast \coloneqq dt\vert_{1_H} \,:\, \h \,\longrightarrow\, \g
\end{flalign}
is the differential of the Lie group homomorphism $t$ at the identity $1_H\in H$
and the Lie algebra homomorphism
\begin{subequations}\label{eqn:alpha ast}
\begin{flalign}
\alpha_\ast \coloneqq d\alpha\vert_{1_G}\,:\, \g ~\longrightarrow~\mathrm{Der}(\h)
\end{flalign}
is the differential of the adjunct $\alpha : G\to \Aut(H)$ of $\alpha$ at $1_G\in G$. 
One can equivalently regard $\alpha_\ast$ as a linear map 
\begin{flalign}
\alpha_\ast \,:\, \g\otimes\h  ~\longrightarrow~ \h\quad,
\end{flalign}
\end{subequations}
for which the Lie derivation property reads as
\begin{flalign}
\alpha_\ast\big(x,[y,y^\prime]\big)\,=\, \big[\alpha_\ast(x,y),y^\prime\big] + \big[y,\alpha_\ast(x,y^\prime)\big]\quad,
\end{flalign}
for all $x\in\g$ and $y,y^\prime\in\h$. The two properties \eqref{eqn:CM axioms Lie} 
differentiate to
\begin{subequations} \label{eqn:CM axioms differential}
\begin{flalign}
\label{eqn:CM axiom differential a} t_\ast\big( \alpha_\ast(x, y) \big) \,&=\, \big[x,t_\ast(y)\big]\quad, \\
\label{eqn:CM axiom differential b} \alpha_\ast\big( t_\ast(y), y^\prime \big) \,&=\, [y,y^\prime]\quad,
\end{flalign}
\end{subequations}
for all $x\in\g$ and $y,y^\prime\in\h$.

There are two additional differentials of the smooth action $\alpha : G\times H\to H$
which play an important role in higher gauge theory. With a slight abuse of notation,
we shall denote all these differentials by the same symbol $\alpha_\ast$. 
First, for every $g\in G$, one has a Lie group homomorphism $\alpha_g \coloneqq \alpha(g,\,\cdot\,) : H\to H$
of which one can take the differential $d\alpha_g\vert_{1_H} :\h\to \h$ at $1_H\in H$.
Allowing $g\in G$ to vary, one obtains the map
\begin{flalign}\label{eqn:alpha ast G}
\alpha_\ast \coloneqq d\alpha_{(\,\cdot\,)}\vert_{1_H}\,:\,G\times\h\,\longrightarrow\,\h\quad.
\end{flalign}
Second, for every $h\in H$, one has a smooth map $\tilde{\alpha}_h\coloneqq \alpha(\,\cdot\,,h)\,h^{-1} : G\to H$
which preserves the identity elements, i.e.\ $\tilde{\alpha}_h(1_G)=1_H$, 
but not necessarily the group multiplications.
Taking the differential at $1_G\in G$ defines a linear map $d\tilde{\alpha}_{h}\vert_{1_G} : \g\to \h$. 
Allowing $h\in H$ to vary, one obtains the map
\begin{flalign}\label{eqn:alpha ast H}
\alpha_\ast \coloneqq d\tilde{\alpha}_{(\,\cdot\,)}\vert_{1_G}\,:\,\g\times H\,\longrightarrow\,\h\quad.
\end{flalign}
Note that the three maps in \eqref{eqn:alpha ast}, \eqref{eqn:alpha ast G} and
\eqref{eqn:alpha ast H} can be distinguished from their source,
so denoting all of them by the same symbol will likely cause no confusion.

\subsection{\label{subsec:2grpd}Higher gauge fields and gauge transformations}
Given any manifold $X$ and crossed module of Lie groups $(G,H,t,\alpha)$,
we will always consider the corresponding trivial principal $2$-bundle over $X$. 
The following definition is from \cite[Appendix A.1]{Waldorf2}.
\begin{definition}\label{def:connections}
Let $X$ be a manifold and $(G,H,t,\alpha)$ crossed module of Lie groups.
\begin{itemize}
\item[(a)] A \textit{connection} is a pair $(A,B)$ consisting of 
a $\g$-valued $1$-form $A\in\Omega^1(X,\g)$ and an $\h$-valued 
$2$-form $B\in \Omega^2(X,\h)$. 

\item[(b)] A \textit{gauge transformation} 
is a pair $(g, \gamma)$ consisting of a $G$-valued 
smooth function $g\in C^\infty(X,G)$ and an $\h$-valued $1$-form $\gamma\in\Omega^1(X,\h)$.
It transforms a connection $(A,B)$ to the connection ${}^{(g,\gamma)}(A,B)$ specified by
\begin{subequations}\label{eqn:gauge transformation}
\begin{flalign}
\label{eqn:gauge transformation b} {}^{(g,\gamma)}A\,&\coloneqq\, g \,A\, g^{-1} - \dd g \,g^{-1} - t_\ast (\gamma)\quad,\\
\label{eqn:gauge transformation c} {}^{(g,\gamma)}B\,&\coloneqq\, \alpha_\ast(g, B) - F(\gamma) 
- \alpha_\ast\big({}^{(g,\gamma)}A, \gamma\big)  \quad,
\end{flalign}
\end{subequations}
where $F(\gamma) \coloneqq \dd \gamma + \tfrac{1}{2} [\gamma, \gamma] \in \Omega^2(X, \h)$.

\item[(c)] We denote by $\mathrm{Con}_{(G,H,t,\alpha)}(X)$ the groupoid whose objects
are all connections $(A,B)$ and whose morphisms $(g,\gamma) : (A,B)\to {}^{(g,\gamma)}(A,B)$ 
are all gauge transformations between connections. The composition of two morphisms 
$(g_1,\gamma_1): (A,B)\to {}^{(g_1,\gamma_1)}(A,B)$ and $(g_2,\gamma_2) :{}^{(g_1,\gamma_1)}(A,B)\to
{}^{(g_2,\gamma_2)}{}^{(g_1,\gamma_1)}(A,B)$ is defined by
\begin{flalign}
(g_2, \gamma_2) \, (g_1, \gamma_1) \,\coloneqq\, \big( g_2 \,g_1,\, \gamma_2 + \alpha_\ast(g_2, \gamma_1)  \big)\quad,
\end{flalign}
for all $g_1, g_2 \in C^\infty(X, G)$ and $\gamma_1, \gamma_2 \in \Omega^1(X, \h)$, and the identity
morphisms are  $(1_G,0) : (A,B)\to (A,B)$. The inverse of a 
morphism $(g,\gamma) : (A,B)\to {}^{(g,\gamma)}(A,B)$ is given explicitly by
\begin{flalign}
(g,\gamma)^{-1}\,=\,\big(g^{-1}, - \alpha_\ast(g^{-1},\gamma)\big)\quad.
\end{flalign}
\end{itemize}
\end{definition}

\begin{remark}\label{rem:2gauge}
For the purpose of our paper, it is important that 
we do \textit{not} demand the fake-flatness condition  
$\mathsf{fcurv}(A,B) \,:=\, F(A) - t_\ast(B) \,:=\, \dd A +\tfrac{1}{2}[A,A] - t_\ast (B)=0$
because this would degenerate the action functionals studied in the later sections.
The theory of non-fake-flat connections is unfortunately not yet well
understood, especially when it comes to their $2$-categorical aspects. 
In the fake-flat case,
there also exists a concept of $2$-gauge transformations $a : (g,\gamma)\Rightarrow {}^{a}(g,\gamma)$ 
between gauge transformations, which are parametrized by $H$-valued smooth functions $a\in C^\infty(X,H)$
and transform according to
\begin{subequations}
\begin{flalign}
{}^a g \,&:=\, t(a)\,g\quad,\\
{}^a\gamma \,&:=\,a\,\gamma\,a^{-1} - \dd a\,a^{-1} - \alpha_\ast\big({}^{(g,\gamma)}A,a\big)\quad.
\end{flalign}
\end{subequations}
These $2$-gauge transformations are only well-defined in the non-fake-flat case provided
that the rather unnatural and highly restrictive constraint $\alpha_\ast\big(\mathsf{fcurv}(A,B),a\big)=0$ holds true.
Because of these issues, we will neglect such $2$-gauge transformations in our constructions
and consider only the ordinary groupoid $\mathrm{Con}_{(G,H,t,\alpha)}(X)$
of connections and gauge transformations instead of a potential $2$-groupoid refinement
involving also $2$-gauge transformations.
\end{remark}

\subsection{\label{subsec:2hol}Higher parallel transport}
Given any connection $(A,B)$ as in Definition \ref{def:connections},
one can define parallel transports along $1$-dimensional paths
$\gamma : [0,1]\to X\,,s\mapsto \gamma(s)$ and also along $2$-dimensional surfaces 
$\Gamma: [0,1]^2\to X\,,~(s,u)\mapsto \Gamma(s,u)$. The parallel transport
along a $1$-dimensional path $\gamma$ is given by evaluating the path-ordered exponential
\begin{subequations}
\begin{flalign}
\mathsf{g}_{\gamma}(s)\,=\,\mathcal{P}\exp\bigg(\int_0^s A\big(\tfrac{d}{ds^\prime}\gamma(s^\prime)\big)\,\dd s^\prime\bigg)\,\in\,G
\end{flalign}
at $s=1$, where $\tfrac{d}{ds^\prime}\gamma(s^\prime)\in T_{\gamma(s^\prime)}X$ 
denotes the tangent vector. Recall that the path-ordered exponential is defined as the solution 
of the differential equation
\begin{flalign}
\tfrac{d}{ds}\mathsf{g}_{\gamma}(s)\,=\, \mathsf{g}_{\gamma}(s)~A\big(\tfrac{d}{ds}\gamma(s)\big)
\end{flalign}
\end{subequations}
for the initial condition $\mathsf{g}_{\gamma}(0)=1_G$.
The parallel transport along a $2$-dimensional surface $\Gamma$ is given by evaluating the surface-ordered exponential
\begin{subequations}
\begin{flalign}
\mathsf{h}_{\Gamma}(s,u)\,=\,\mathcal{S}\exp\bigg(\int_0^u\int_0^s B\Big(\tfrac{\partial}{\partial s^\prime}\Gamma(s^\prime,u^\prime),\tfrac{\partial}{\partial u^\prime}\Gamma(s^\prime,u^\prime)\Big)\,\dd s^\prime\,\dd u^\prime\bigg)\,\in\,H
\end{flalign}
at $(s,u)=(1,1)$. The latter is defined as the solution of the differential equation
\begin{flalign}
\tfrac{\partial}{\partial u}\mathsf{h}_{\Gamma}(s,u)\,=\, \mathsf{h}_{\Gamma}(s,u)~\int_0^s\alpha_\ast\bigg(\mathsf{g}_{\gamma^0}(u)\,\mathsf{g}_{\gamma_u}(s^\prime),\,
B\Big(\tfrac{\partial}{\partial s^\prime}\Gamma(s^\prime,u),\tfrac{\partial}{\partial u}\Gamma(s^\prime,u)\Big)\bigg)
\,\dd s^\prime
\end{flalign}
\end{subequations}
for the initial condition $\mathsf{h}_{\Gamma}(s,0)=1_H$, for all $s\in[0,1]$,
where we used the notation $\Gamma(s,u)=\gamma^s(u)=\gamma_u(s) $.

In \cite{SchreiberWaldorf,Joao,Waldorf2}, it is shown that in the case where
the connection $(A,B)$ is fake-flat, i.e.\ the fake-curvature
$F(A) - t_\ast(B) =0$ vanishes, 
the $1$- and $2$-dimensional parallel transports assemble into a $2$-functor from the path $2$-groupoid 
of the manifold $X$ to the classifying space of the structure Lie $2$-group $(G,H,t,\alpha)$.
This means that parallel transports satisfy suitable $1$- and $2$-dimensional
composition properties, which allow one to compute them also along non-trivial 
shapes such as $2$-spheres or $2$-tori.

In our context of $3$-dimensional integrable field theories, as
discussed later in Section \ref{sec:integrable}, the $2$-dimensional
parallel transport of a fully flat connection $(A,B)$ on a $3$-dimensional spacetime $X=\bbR^3$ 
gives rise to conserved charges which are localized on fixed-time surfaces
$\{t\}\times \bbR^2\subset X$. Let us explain this crucial point more explicitly.
In order to ensure that the relevant integrals
exist, we assume that the connection $(A,B)$ has spacelike compact support on $X$,
i.e.\ its restriction $(A,B)\vert_t$ to each fixed-time surface $\{t\}\times\bbR^2\subset X$ is compactly supported.
Denoting by $\Gamma_t : [0,1]^2\to \{t\}\times\bbR^2\subset X$ any fixed-time 
$2$-dimensional surface which contains the support of $(A,B)\vert_t$, we consider
the $2$-dimensional parallel transport
\begin{flalign}\label{eqn:surfacetransport}
\mathsf{h}_{\Gamma_t}\,=\,\mathcal{S}\exp\bigg(\int_0^1\int_0^1 B\Big(\tfrac{\partial}{\partial s^\prime}\Gamma_t(s^\prime,u^\prime),\tfrac{\partial}{\partial u^\prime}\Gamma_t(s^\prime,u^\prime)\Big)\,\dd s^\prime\,\dd u^\prime\bigg)\,\in\,H\quad.
\end{flalign}
In the case where $(A,B)$ is fully flat, i.e.\ the fake-curvature
$F(A) - t_\ast(B) =0$ and the $3$-form curvature
$\dd B + \alpha_\ast(A,B) =0$ both vanish,
this parallel transport is time-independent, i.e.\
\begin{flalign}
\mathsf{h}_{\Gamma_t} \,=\,\mathsf{h}_{\Gamma_{t^\prime}}
\end{flalign}
for all $t,t^\prime\in\bbR$. (This statement follows from 
the same argument which proves thin homotopy invariance of $2$-dimensional parallel transports,
see e.g.\ \cite[Appendix A.3]{SchreiberWaldorf} or \cite[Theorem 2.30]{Joao}.)
This implies that $\mathsf{h}_{\Gamma_t}\in H$ is a conserved charge.
In the case of a $3$-dimensional integrable field theory (see Section \ref{sec:integrable}), 
the fully flat higher Lax connection $(A,B)$ further depends meromorphically on 
an auxiliary Riemann surface $C$, and hence so does the conserved charge $\mathsf{h}_{\Gamma_t}$.
Laurent expansion then leads to an infinite tower of conserved charges which are localized
on the fixed-time surfaces $\{t\}\times \bbR^2\subset X$.
\begin{remark}
We would like to point out a potential relationship between the higher parallel transports
from this subsection and the concept of generalized (higher-form) symmetries from
\cite{Gaiotto:2014kfa,Benini:2018reh}. The conserved charges 
\eqref{eqn:surfacetransport} constructed from $2$-dimensional
parallel transports of a fully flat connection $(A,B)$ are localized in codimension $1$
of spacetime $X=\bbR^3$, so they correspond to $0$-form symmetries
according the terminology of generalized symmetries. Considering 
loops in a fixed-time surface, i.e.\ $1$-dimensional paths 
$\gamma: [0,1]\to \{t\}\times\bbR^2\subset X$ such that $\gamma(0)=\gamma(1)$,
we can consider also $1$-dimensional parallel transports $\mathsf{g}_\gamma\in G$
which are localized in codimension $2$ of spacetime. It is important
to stress that the fake-flatness condition $F(A) - t_\ast(B)=0$
does in general \textit{not} imply that these codimension $2$ charges
are conserved under deformations of the loop, so they are in general \textit{not}
$1$-form symmetries in the usual sense of generalized symmetries.
However, the violation of the conservation of $\mathsf{g}_\gamma\in G$
under deformations of the loop is controlled by $2$-dimensional
parallel transports since the curvature $F(A)=t_\ast(B)$ of the $1$-form 
$A$ is determined via fake-flatness in terms of the $2$-form $B$. This seems 
to indicate that we are dealing with a system of $0$-form and $1$-form symmetries
which have a non-trivial interplay with each other. It would be interesting to 
make the latter point more precise by identifying the algebraic structures 
underlying this interplay. We would like to note that for a crossed module
with $t_\ast=0$, fake-flatness implies that $F(A)=0$, hence the 
codimension $2$ charges $\mathsf{g}_\gamma\in G$ decouple from the codimension
$1$ charges $\mathsf{h}_{\Gamma_t}\in H$ and, as a consequence, are strictly conserved.
\end{remark}

\subsection{\label{subsec:2CSform}The $2$-Chern-Simons $4$-form}
Let $X$ be a manifold of dimension $n\geq 4$ and $(G,H,t,\alpha)$ a crossed module of Lie groups.
Suppose that we are given a non-degenerate pairing
\begin{subequations}\label{eqn: pairing}
\begin{flalign}
\ip{\,\cdot\,}{\,\cdot\,} \,:\, \g\otimes \h \,\longrightarrow\,\bbC
\end{flalign} 
on the underlying crossed module of Lie algebras $(\g,\h,t_\ast,\alpha_\ast)$
which is $G$-invariant, i.e.\
\begin{flalign}\label{eqn: pairing G invariant}
\ip{g\,x\,g^{-1}}{\alpha_\ast(g,y)}\,=\,\ip{x}{y}\quad,
\end{flalign}
for all $g\in G$, $x\in\g$ and $y\in \h$, and satisfies
the symmetry property
\begin{flalign}\label{eqn: pairing symmetry}
\ip{t_\ast(y)}{y^\prime} \,= \, \ip{t_\ast(y^\prime)}{y}\quad,
\end{flalign}
\end{subequations}
for all $y,y^\prime\in\h$.
Note that the $G$-invariance property
\eqref{eqn: pairing G invariant} differentiates to 
the $\g$-invariance property
\begin{flalign}\label{eqn: pairing g invariant}
\ip{[x,x^\prime]}{y}+ \ip{x^\prime}{\alpha_\ast(x,y)} \,=\,0\quad,
\end{flalign}
for all $x,x^\prime\in\g$ and $y\in \h$.

We will now construct from these data a $4$-form on $X$ which is a 
higher-dimensional generalization of the usual Chern-Simons $3$-form in ordinary 
gauge theory. As a first step, let us observe that there exists a double complex
\begin{flalign}
\begin{gathered}
\xymatrix{
\stackrel{(0,0)}{\Omega^0(X,\g)} \ar[r]^-{\dd}& \stackrel{(0,1)}{\Omega^1(X,\g)} \ar[r]^-{\dd}&\cdots \ar[r]^-{\dd}&\stackrel{(0,n)}{\Omega^n(X,\g)}\\
\ar[u]^-{t_\ast}\stackrel{(-1,0)}{\Omega^0(X,\h)} \ar[r]^-{\dd}& \ar[u]^-{t_\ast}\stackrel{(-1,1)}{\Omega^1(X,\h)} \ar[r]^-{\dd}&\cdots \ar[r]^-{\dd}& \ar[u]^-{t_\ast}\stackrel{(-1,n)}{\Omega^n(X,\h)}
}
\end{gathered}
\end{flalign}
whose horizontal differential is the de Rham differential $\dd$ and vertical differential is the linear map $t_\ast$
from the crossed module of Lie algebras $(\g,\h,t_\ast,\alpha_\ast)$. The bi-degrees are as indicated in the parentheses.
Totalizing this double complex leads to the cochain complex
\begin{flalign}
L\,\coloneqq\,\left(
{\footnotesize\xymatrix@C=2em{
{\begin{array}{c}
0\\
\oplus\\
\Omega^0(X,\h)
\end{array}}
\ar[r]^-{\dd_L}
&
{\begin{array}{c}
\Omega^0(X,\g)\\
\oplus\\
\Omega^1(X,\h)
\end{array}}
\ar[r]^-{\dd_L}
&
~\cdots~
\ar[r]^-{\dd_L}
&
{\begin{array}{c}
\Omega^{n-1}(X,\g)\\
\oplus\\
\Omega^n(X,\h)
\end{array}}
\ar[r]^-{\dd_L}
&
{\begin{array}{c}
\Omega^n(X,\g)\\
\oplus\\
0
\end{array}}
}}
\right)
\end{flalign}
concentrated in degrees $\{-1,0,\dots,n\}$.
More explicitly, a cochain of degree $p$ in $L$
is a pair $\ell = \omega\oplus \eta\in L^p = \Omega^p(X,\g)\oplus\Omega^{p+1}(X,\h)$
of Lie algebra-valued differential forms and the differential $\dd_L$ reads explicitly as
\begin{flalign}
\dd_L\ell \,\coloneqq\,\dd_L\big(\omega\oplus \eta\big)\,\coloneqq\,
\big(\dd \omega + (-1)^p \,t_\ast(\eta)\big)\oplus \dd\eta\quad.
\end{flalign}
Note that the degree $1$ cochains $\A = A \oplus B\in L^1 =\Omega^1(X,\g)\oplus\Omega^2(X,\h)$
in this complex are precisely the connections on the trivial principal $2$-bundle associated with $(G,H,t,\alpha)$
from Definition \ref{def:connections}.

Using the additional structures from the crossed module of Lie algebras 
$(\g,\h,t_\ast,\alpha_\ast)$, one can endow the cochain complex $L$
with the structure of a dg-Lie algebra. The Lie bracket reads explicitly
as
\begin{flalign}
\big[\ell,\ell^\prime\big]_L\, \coloneqq\, [\omega,\omega^\prime] \oplus\Big(\alpha_\ast(\omega,\eta^\prime) -
(-1)^{p p^\prime} \alpha_\ast(\omega^\prime,\eta)\Big)\quad,
\end{flalign}
for all $\ell = \omega\oplus \eta\in L^p$ and $\ell^\prime = \omega^\prime\oplus \eta^\prime\in L^{p^\prime}$.
Finally, using also the non-degenerate pairing \eqref{eqn: pairing}, one obtains
a differential form-valued cyclic structure 
\begin{flalign}
\ip{\,\cdot\,}{\,\cdot\,}_L\,:\,L\otimes L \,\longrightarrow\,\Omega^{\bullet+1}(X)
\end{flalign}
of degree $1$ on the dg-Lie algebra $L$. This reads explicitly as
\begin{flalign}
\ip{\ell}{\ell^\prime}_L\, \coloneqq\,\ip{\omega\oplus\eta}{\omega^\prime\oplus\eta^\prime}_L\, \coloneqq\,
\ip{\omega}{\eta^\prime}+ (-1)^{pp^\prime}\,\ip{\omega^\prime}{\eta}\quad,
\end{flalign}
for all $\ell = \omega\oplus \eta\in L^p$ and $\ell^\prime = \omega^\prime\oplus \eta^\prime\in L^{p^\prime}$.

The following definition is motivated by the construction of higher Chern-Simons actions
in terms of Maurer-Cartan theory, see e.g.\ \cite[Section 5.2]{Jurco}.
\begin{definition}\label{def:2CS 4 form}
The \textit{$2$-Chern-Simons $4$-form} associated to a connection 
$\A=A\oplus B \in L^1 = \Omega^1(X,\g)\oplus \Omega^2(X,\h)$ is defined as
\begin{flalign}\label{eqn:2CS 4 form def}
\CS(\A)\,\coloneqq\, \CS(A,B) \,\coloneqq\, \ip{\A}{\tfrac{1}{2}\dd_L \A + \tfrac{1}{3!}[\A,\A]_L}_L\,\in\,\Omega^4(X)\quad.
\end{flalign}
\end{definition}

\begin{lemma}
The $2$-Chern-Simons $4$-form \eqref{eqn:2CS 4 form def} 
can be expanded in components as
\begin{flalign}\label{eqn:2CS 4 form}
\CS(A, B) \,=\, \ip{F(A) - \tfrac{1}{2} t_\ast (B)}{B}- \tfrac{1}{2} \dd \ip{A}{B}\quad,
\end{flalign}
where we recall that $F(A) = \dd A + \tfrac{1}{2} [A,A] \in \Omega^2(X, \g)$.
\end{lemma}
\begin{proof}
We have
\begin{flalign}
\nonumber \CS(A,B) \,&=\, \ip{ A \oplus B }{ \tfrac{1}{2} \big((\dd A - t_\ast (B)) \oplus \dd B\big) 
+ \tfrac{1}{3!} \big( [A, A] \oplus 2 \,\alpha_\ast(A, B) \big)}_L\\
\nonumber \,&=\, \tfrac{1}{2} \ip{ A }{ \dd B } + \tfrac{1}{3} \ip{ A }{ \alpha_\ast(A,B) }
+ \tfrac{1}{2} \ip{ \dd A - t_\ast (B) }{ B } + \tfrac{1}{3!} \ip{ [A,A] }{ B }\\
\nonumber\, &=\, \tfrac{1}{2} \ip{ A }{ \dd B } + \tfrac{1}{2} \ip{ \dd A  + [A, A] - t_\ast (B) }{ B }\\
\, &=\,  \ip{ \dd A + \tfrac{1}{2} [A,A] - \tfrac{1}{2} t_\ast(B) }{ B } - \tfrac{1}{2} \dd \ip{ A }{ B }\quad,
\end{flalign}
where in the third step we used the $\g$-invariance property \eqref{eqn: pairing g invariant}
and the fact that $A$ is a $1$-form to write $\ip{A}{\alpha_\ast(A,B)}=\ip{[A,A]}{B}$, 
and in last step we used $\ip{A}{\dd B} = \ip{\dd A}{B} - \dd \ip{A}{B}$.
\end{proof}

We will require later the following result about the transformation behavior 
of the $2$-Chern-Simons $4$-form under the gauge transformations
from Definition \ref{def:connections}.
\begin{proposition}\label{propo:2CS 4form gauge}
The $2$-Chern-Simons $4$-form \eqref{eqn:2CS 4 form} transforms 
under gauge transformations $(g,\gamma) : (A,B)\to {}^{(g,\gamma)}(A,B)$,
for $g\in C^\infty(X,G)$ and $\gamma\in\Omega^1(X,\h)$, as
\begin{flalign}
\CS\big({}^{(g,\gamma)}(A,B)\big)\,=\, \CS(A,B)  -\tfrac{1}{2}\dd\Big(&\ip{g A g^{-1}}{F(\gamma)} 
+ \ip{t_\ast(\gamma)}{\dd \gamma+\tfrac{1}{3}[\gamma,\gamma]} \\
\nonumber &- \ip{\dd g \,g^{-1}+t_\ast(\gamma)}{\alpha_\ast(g,B) + F(\gamma)}\Big)\quad,
\end{flalign}
where we recall that $F(\gamma) =\dd\gamma +\tfrac{1}{2}[\gamma,\gamma] \in \Omega^2(X, \h)$.
\end{proposition}
\begin{proof}
This is a lengthy but straightforward computation, 
so we just highlight the main steps. For the first term of
the $2$-Chern-Simons $4$-form \eqref{eqn:2CS 4 form},
one uses the explicit form of the gauge transformations 
\eqref{eqn:gauge transformation b} and \eqref{eqn:gauge transformation c}
to show that $F\big({}^{(g,\gamma)}A\big) - \tfrac{1}{2} t_\ast \big({}^{(g,\gamma)}B\big) 
= g \,\big(F(A) - \tfrac{1}{2} t_\ast (B)\big)\, g^{-1} - 
\tfrac{1}{2} t_\ast \big( F(\gamma) + \alpha_\ast\big({}^{(g,\gamma)}A, \gamma\big) \big)$. 
From this it then follows that
\begin{flalign}
 \big\langle F\big({}^{(g,\gamma)}A\big)& - 
\tfrac{1}{2} t_\ast \big({}^{(g,\gamma)}B\big) , {}^{(g,\gamma)}B \big\rangle
= \big\langle F(A) - \tfrac{1}{2} t_\ast(B), B \big\rangle \\
\nonumber &- \big\langle F\big({}^{(g,\gamma)}A\big) + 
\tfrac{1}{2} t_\ast \big( F(\gamma) + \alpha_\ast\big({}^{(g,\gamma)}A, \gamma\big) \big), 
F(\gamma) + \alpha_\ast\big({}^{(g,\gamma)}A, \gamma\big) \big\rangle\quad.
\end{flalign}
The second term on the right-hand side is found to be exact by making 
repeated use of the properties \eqref{eqn:CM axioms differential} and 
\eqref{eqn: pairing g invariant}, and the Jacobi identities for $\g$ and $\h$.
Explicitly, one finds
\begin{flalign}
\big\langle F\big({}^{(g,\gamma)}A\big) 
&- \tfrac{1}{2} t_\ast \big({}^{(g,\gamma)}B\big), {}^{(g,\gamma)}B \big\rangle
= \big\langle F(A) - \tfrac{1}{2} t_\ast (B), B \big\rangle \\
\nonumber &- \tfrac{1}{2} 
\dd \Big( \big\langle {}^{(g,\gamma)}A, 2 F(\gamma) + \alpha_\ast\big({}^{(g,\gamma)}A, \gamma\big) \big\rangle 
+ \big\langle t_\ast(\gamma), \dd \gamma + \tfrac{1}{3} [\gamma, \gamma] \big\rangle \Big)\quad.
\end{flalign}
For the second term of the $2$-Chern-Simons $4$-form \eqref{eqn:2CS 4 form}, one finds
\begin{flalign}
- \tfrac{1}{2} \dd &\big\langle {}^{(g,\gamma)}A, {}^{(g,\gamma)}B \big\rangle = 
- \tfrac{1}{2} \dd \langle A, B \rangle \\
\nonumber &+ \tfrac{1}{2} \dd \Big( \big\langle \dd g \,g^{-1} + t_\ast (\gamma),
\alpha_\ast(g, B) \big\rangle + \big\langle {}^{(g,\gamma)}A, F(\gamma) + \alpha_\ast\big({}^{(g,\gamma)}A, \gamma\big) \big\rangle \Big)\quad.
\end{flalign}
The result now follows by combining the above transformation properties of 
the first and second term in the $2$-Chern-Simons $4$-form \eqref{eqn:2CS 4 form}.
\end{proof}


\section{\label{sec:simplepoles}$5d$ $2$-Chern-Simons theory with simple poles}
In this section we define and analyze a $5$-dimensional generalization
of $4$-dimensional semi-holomorphic Chern-Simons theory 
\cite{Costello:2013sla,Costello:2017dso,Costello:2018gyb,Costello:2019tri}.
Our notations and conventions follow \cite{Benini:2020skc}.
In order to simplify our presentation, we consider in this section
first the special case where $\omega$ is a meromorphic $1$-form on the Riemann sphere $\CP$
which has only simple poles and postpone the more involved 
case of higher-order poles to Section \ref{sec:higherpoles}.
We denote by $\bm{z}\subset \CP$ the set of poles and by $\bm{\zeta}\subset \CP$
the set of zeros of $\omega$. We assume that $\omega$ has at least one zero,
i.e.\ $\vert \bm{\zeta}\vert \geq 1$.

The field theory we consider is defined 
on the $5$-dimensional manifold
\begin{flalign}
X\,\coloneqq\, M\times C\quad,
\end{flalign}
where $M= \bbR^3$ is the $3$-dimensional Cartesian space, 
which we interpret as spacetime, and $C \coloneqq \CP\setminus\bm{\zeta}$ 
is the Riemann sphere with all zeros of $\omega$ removed. 
We choose a global coordinate $z : C\to \bbC$ on the factor $C$,
which exists because it is assumed that $\vert \bm{\zeta}\vert \geq 1$.
 
\subsection{\label{subsec:action simple pole}Action} 
Let $(G,H,t,\alpha)$ be a crossed module of Lie groups 
endowed with a non-degenerate invariant pairing $\ip{\,\cdot\,}{\,\cdot\,}:\g\otimes\h\to \bbC$ 
as in \eqref{eqn: pairing}. We assume throughout the whole paper
that both $G$ and $H$ are connected and simply-connected.
Using the $2$-Chern-Simons $4$-form from 
Definition \ref{def:2CS 4 form}, we define the action 
\begin{flalign}\label{eqn:5d CS action}
S_\omega(A,B)\,\coloneqq\, \frac{\ii}{2 \pi} \int_X \omega\wedge\CS(A,B)
\end{flalign}
on the set of connections $(A,B) \in\Omega^1(X,\g)\times \Omega^2(X,\h)$,
where with a slight abuse of notation we denote the pullback of the meromorphic
$1$-form $\omega$ along the projection $X = M\times C\to C$ by the same symbol.

Note that the components of $A$ and $B$ with a leg along $\dd z$ do not 
contribute to the action \eqref{eqn:5d CS action} because $\omega = \varphi\,\dd z$ 
is a meromorphic $1$-form. To remove these non-dynamical fields 
from the theory, we consider in what follows connections which are modeled on the quotient
\begin{subequations}\label{eqn:dR quotient}
\begin{flalign}
\overline{\Omega}^\bullet(X) \,\coloneqq\, \Omega^\bullet(X)\big/(\dd z)
\end{flalign}
of the de Rham calculus on $X$, endowed with the differential
\begin{flalign}
\overline{\dd}\,\coloneqq\,\dd_M + \overline{\partial}
\end{flalign}
\end{subequations}
given by the sum of the differential $\dd_M$ along the factor $M$ and 
the Dolbeault differential $\overline{\partial}$ along the factor $C$
of the product manifold $X = M\times C$.
We denote by
\begin{flalign}\label{eqn:overlineCon}
\overline{\mathrm{Con}}_{(G,H,t,\alpha)}(X)
\end{flalign}
the analogue of the groupoid from Definition \ref{def:connections}
where connections and gauge transformations are modeled on the quotient 
de Rham calculus $(\overline{\Omega}^\bullet(X),\overline{\dd})$.

We observe that the integrand of the action \eqref{eqn:5d CS action} is singular at the $3$-dimensional defect
\begin{flalign}\label{eqn:surface defect}
D\,\coloneqq\,M\times \bm{z} \,=\,\bigsqcup_{x\in \bm{z}} \big(M\times \{x\}\big)\,\subset \,X\
\end{flalign}
which is localized at the poles $x\in \bm{z}$ of $\omega$.
In the present case of simple poles, it follows by the same argument as
in \cite[Lemma 2.1]{Benini:2020skc} that these singularities
are locally integrable near each component $M\times\{x\}\subset X$ of the defect,
hence the action \eqref{eqn:5d CS action} is well-defined. 
(In the case of higher-order poles, the integral requires a regularization, 
see Section \ref{sec:higherpoles} for the details.)

\subsection{\label{subsec:gauge simple}Gauge transformations}
In analogy to the $4$-dimensional case, the action \eqref{eqn:5d CS action} 
of $5$-dimensional semi-holomorphic $2$-Chern-Simons theory is not invariant 
under arbitrary gauge transformations $(g,\gamma) : (A,B) \to {}^{(g,\gamma)}(A,B)$
as in Definition \ref{def:connections} (with the differential $\dd$ replaced
by $\overline{\dd}$ from the quotient de Rham calculus 
\eqref{eqn:dR quotient}), where $g\in C^\infty(X,G)$ and $\gamma\in \overline{\Omega}^1(X,\h)$.
Indeed, using Proposition \ref{propo:2CS 4form gauge}, we compute
\begin{flalign}\label{eqn:2CS transformation}
S_\omega\big({}^{(g,\gamma)}(A,B)\big)\,=\, S_\omega(A,B) -\frac{\ii}{4 \pi}\int_X \omega\wedge \overline{\dd}\Big(&\ip{g A g^{-1}}{\overline{F}(\gamma)} 
+ \ip{t_\ast(\gamma)}{\overline{\dd} \gamma+\tfrac{1}{3}[\gamma,\gamma]} \\
\nonumber &- \ip{\overline{\dd} g \,g^{-1}+t_\ast(\gamma)}{\alpha_\ast(g,B)+\overline{F}(\gamma)}\Big)\quad,
\end{flalign}
where $\overline{F}(\gamma) :=\overline{\dd}\gamma + \tfrac{1}{2}[\gamma,\gamma]\in \overline{\Omega}^2(X,\h)$
denotes the curvature with respect to the quotient de Rham calculus \eqref{eqn:dR quotient},
and note that the second term does in general not vanish. 
The violation of gauge invariance is due to the poles of $\omega$ and hence
it is localized at the defect \eqref{eqn:surface defect}. 
To state this observation precisely, we have to introduce some more notations
and terminology. Let us denote the inclusions of the defect into $X$ by
\begin{flalign}
\iota_x \,:\, M\times\{x\} \,\longhookrightarrow \, X\quad,\qquad \bm{\iota}\,\coloneqq\,\bigsqcup_{x\in\bm{z}}\iota_x\,:\, D\,\longhookrightarrow\,X\quad.
\end{flalign}
Pulling back functions or differential forms
along these inclusions defines maps
\begin{subequations}\label{eqn:iota pullback maps}
\begin{flalign}
\bm{\iota}^\ast\,:\,C^\infty(X,N)\,&\longrightarrow\, C^\infty(D,N)\,\cong\, C^\infty(M,N^{\bm{z}})\quad,\\
\bm{\iota}^\ast\,:\, \overline{\Omega}^q(X,V)\,&\longrightarrow\,\Omega^q(D,V) \,\cong\, \Omega^q(M,V^{\bm{z}})\quad,
\end{flalign}
\end{subequations}
where $N$ is any smooth manifold (e.g.\ one of the Lie groups $G$ or $H$)
and $V$ is any vector space (e.g.\ one of the Lie algebras $\g$ or $\h$).
Here $(\,\cdot\,)^{\bm{z}} := \prod_{x\in \bm{z}}(\,\cdot\,)$ 
denotes the product over all poles. We extend the pairing 
$\ip{\,\cdot\,}{\,\cdot\,} : \g\otimes\h\to\bbC$ to these products by setting
\begin{subequations}\label{eqn:defectpairing simple}
\begin{flalign}
\ipp{\,\cdot\,}{\,\cdot\,}_{\omega}\,:\, \g^{\bm{z}}\otimes\h^{\bm{z}}\,\longrightarrow\bbC~~,\quad
\mathcal{X}\otimes \mathcal{Y}\,\longmapsto\,\ipp{\mathcal{X}}{\mathcal{Y}}_\omega\,\coloneqq\,
\sum_{x\in\bm{z}} k^x\ip{\mathcal{X}^x}{\mathcal{Y}^x}\quad,
\end{flalign}
where $\mathcal{X} = (\mathcal{X}^x)_{x\in{\bm{z}}}\in\g^{\bm{z}}$,
$\mathcal{Y} = (\mathcal{Y}^x)_{x\in{\bm{z}}}\in\h^{\bm{z}}$ and 
the coefficients 
\begin{flalign}
k^x\,:=\,\mathrm{Res}_{x}(\omega)\, \in\, \bbC
\end{flalign} 
\end{subequations} 
are the residues of $\omega$ at its poles $x\in \bm{z}$.
\begin{proposition}\label{propo:bdy gauge violation}
Under a gauge transformation $(g,\gamma) : (A,B) \to {}^{(g,\gamma)}(A,B)$, 
with $g\in C^\infty(X, G)$ and $\gamma\in\overline{\Omega}^1(X,\h)$, the 
action \eqref{eqn:5d CS action} transforms as
\begin{flalign}
\nonumber S_\omega\big({}^{(g,\gamma)}(A,B)\big) &\,=\,  S_{\omega}(A,B) +\frac{1}{2}\int_M
\Big(\ipp{\bm{\iota}^\ast(g)\, \bm{\iota}^\ast(A)\, \bm{\iota}^\ast(g)^{-1}}{F_M\big(\bm{\iota}^\ast (\gamma)\big)}_{\omega}\\ 
\nonumber &\quad + \ipp{\bm{\iota}^\ast (t_\ast(\gamma))}{\dd_M \bm{\iota}^\ast (\gamma) +\tfrac{1}{3}\big[\bm{\iota}^\ast (\gamma),\bm{\iota}^\ast (\gamma)\big]}_\omega\\
&\quad - \ipp{\dd_M \bm{\iota}^\ast (g) \,\bm{\iota}^\ast (g)^{-1}+\bm{\iota}^\ast (t_\ast(\gamma))}{ \bm{\iota}^\ast\big(\alpha_\ast( g,B) \big)+F_M\big({\bm{\iota}^\ast (\gamma)}\big)}_\omega\Big)\quad,
\end{flalign}
where $F_M\big({\bm{\iota}^\ast (\gamma)}\big) := \dd_M \bm{\iota}^\ast (\gamma) + 
\tfrac{1}{2}\big[\bm{\iota}^\ast (\gamma),\bm{\iota}^\ast (\gamma)\big]\in \Omega^2(M,\h^{\bm{z}})$ 
denotes the curvature and $\dd_M$ the de Rham differential on the $3$-manifold $M$.
\end{proposition}
\begin{proof}
This is a direct consequence of \eqref{eqn:2CS transformation} and 
the Cauchy-Pompeiu integral formula. See \cite[Lemma 2.2]{Benini:2020skc} for more details.
\end{proof}

\subsection{\label{subsec:bdy simple}Boundary conditions at the defect}
The result in Proposition \ref{propo:bdy gauge violation} 
suggests that one has to impose suitable boundary conditions
at the defect $D\subset X$ in order to obtain a gauge-invariant action. 
Note that such boundary conditions must be imposed on both the connections $(A,B)$
and their gauge transformations $(g,\gamma): (A,B)\to {}^{(g,\gamma)}(A,B)$.
We shall focus on a simple class of boundary conditions
which are determined by the choice of a crossed submodule of $(G^{\bm{z}},H^{\bm{z}},t^{\bm{z}},\alpha^{\bm{z}})$
that is isotropic with respect to the pairing $\ipp{\,\cdot\,}{\,\cdot\,}_\omega$ in 
\eqref{eqn:defectpairing simple}. More precisely, this is given by the choice of two Lie subgroups
\begin{subequations}\label{eqn:subcrossed module}
\begin{flalign}
G^{\diamond}\,\subseteq \, G^{\bm{z}}=\prod_{x\in \bm{z}}G\quad,\qquad H^{\diamond}\,\subseteq \, H^{\bm{z}}=\prod_{x\in \bm{z}}H\quad,
\end{flalign}
such that the two structure maps $t^{\bm{z}}=\prod_{x\in\bm{z}}t$ and $\alpha^{\bm{z}}=\prod_{x\in\bm{z}}\alpha$ 
restrict to
\begin{flalign}
t^{\bm{z}}\,:\,H^{\diamond}\,\longrightarrow\,G^{\diamond}\quad,\qquad \alpha^{\bm{z}} \,: \,G^{\diamond}\times H^{\diamond}\,\longrightarrow\,H^{\diamond}\quad.
\end{flalign}
Denoting the associated crossed module of Lie algebras by 
$(\g^{\diamond},\h^{\diamond},t^{\bm{z}}_\ast,\alpha^{\bm{z}}_\ast)$,
the requirement of isotropy means that the restricted pairing vanishes, i.e.
\begin{flalign} \label{eqn:subcrossed module c}
\ipp{\,\cdot\,}{\,\cdot\,}_\omega \big\vert_{\g^{\diamond}\otimes \h^{\diamond}}\,=\,0 \quad.
\end{flalign}
\end{subequations}
We will refer to such a crossed submodule $(G^{\diamond}, H^{\diamond}, t^{\bm{z}}, \alpha^{\bm{z}})$ 
of $(G^{\bm{z}}, H^{\bm{z}}, t^{\bm{z}}, \alpha^{\bm{z}})$ as being isotropic.
\begin{definition}\label{defi:boundarycondition}
The \textit{groupoid of boundary conditioned fields} for the isotropic crossed submodule \eqref{eqn:subcrossed module}
is defined as the subgroupoid $\mathcal{F}^\diamond\subseteq \overline{\mathrm{Con}}_{(G,H,t,\alpha)}(X)$ 
of the groupoid of connections and gauge transformations from \eqref{eqn:overlineCon}
which is specified by the following data:
\begin{itemize}
\item An object in $\mathcal{F}^\diamond$ is a connection 
$(A,B) \in \overline{\Omega}^1(X,\g)\times \overline{\Omega}^2(X,\h)$
which satisfies the boundary condition $\bm{\iota}^\ast(A,B)\in \Omega^1(M,\g^\diamond)\times \Omega^2(M,\h^\diamond)$.

\item A morphism in $\mathcal{F}^\diamond$ is a gauge transformation
$(g,\gamma): (A,B)\to {}^{(g,\gamma)}(A,B)$, with $g\in C^\infty(X,G)$ and $\gamma\in\overline{\Omega}^1(X,\h)$,
which satisfies the boundary condition 
$ \bm{\iota}^\ast(g,\gamma)\in C^\infty(M,G^\diamond)\times \Omega^1(M,\h^\diamond)$.
\end{itemize}
\end{definition}

\begin{proposition}\label{propo:action strict pullback}
For any choice of isotropic crossed submodule 
$(G^{\diamond}, H^{\diamond}, t^{\bm{z}}, \alpha^{\bm{z}})$ as in \eqref{eqn:subcrossed module},
the action \eqref{eqn:5d CS action} restricts to a gauge-invariant function
$S_\omega : \mathcal{F}^\diamond\to \bbC$ on the subgroupoid 
$\mathcal{F}^\diamond\subseteq \overline{\mathrm{Con}}_{(G,H,t,\alpha)}(X)$
of boundary conditioned fields.
\end{proposition}
\begin{proof}
This follows directly from Proposition \ref{propo:bdy gauge violation} and the
isotropy condition \eqref{eqn:subcrossed module c} with respect to $\ipp{\,\cdot\,}{\,\cdot\,}_\omega$.
\end{proof}

\subsection{\label{subsec:edge modes simple}Edge modes}
The groupoid of boundary conditioned fields from 
Definition \ref{defi:boundarycondition}
arises as the \textit{strict} pullback
\begin{flalign}\label{eqn:pullback diagram}
\begin{gathered}
\xymatrix{
\mathcal{F}^{\diamond} \ar@{-->}[r] \ar@{-->}[d]~&~ 
\ar[d]^-{\bm{\iota}^\ast} \overline{\mathrm{Con}}_{(G,H,t,\alpha)}(X)\\
\mathrm{Con}_{(G^{\diamond},H^{\diamond},t^{\bm{z}},\alpha^{\bm{z}})}(M)\ar@{^{(}->}[r] ~&~
\mathrm{Con}_{(G^{\bm{z}},H^{\bm{z}},t^{\bm{z}},\alpha^{\bm{z}})}(M)
}
\end{gathered}
\end{flalign}
in the category of groupoids. Since groupoids naturally form a $2$-category,
it is more appropriate to consider a \textit{homotopy} pullback instead of a strict one.
We shall now spell out an explicit model for the homotopy pullback
of the diagram \eqref{eqn:pullback diagram} and explain how this admits 
an interpretation in terms of edge modes living on the defect $D\subset X$.
\begin{proposition}\label{propo:homotopy pullback}
A model for the homotopy pullback $\mathcal{F}^\diamond_{\mathrm{ho}}$ 
of the diagram \eqref{eqn:pullback diagram} 
is given by the groupoid which is specified by the following data:
\begin{itemize}
\item An object in $\mathcal{F}^\diamond_{\mathrm{ho}}$ is a tuple
$\big((A,B),(k,\kappa)\big)$  consisting of a connection 
$(A,B)\in \overline{\Omega}^1(X,\g)\times \overline{\Omega}^2(X,\h)$
on $X$ and a gauge transformation $(k,\kappa)\in C^\infty(M,G^{\bm{z}})\times\Omega^1(M,\h^{\bm{z}})$
on $M$, such that
\begin{flalign}
{}^{(k,\kappa)}{\bm\iota}^\ast(A,B)\,\in\, \Omega^1(M,\g^\diamond)\times \Omega^2(M,\h^\diamond)\quad.
\end{flalign}

\item A morphism in $\mathcal{F}^\diamond_{\mathrm{ho}}$ is a tuple
\begin{flalign}
\big((g,\gamma),(g^\diamond,\gamma^\diamond)\big)\,:\, \big((A,B),(k,\kappa)\big)~\longrightarrow~
\big({}^{(g,\gamma)}(A,B), (g^\diamond,\gamma^\diamond)\,(k,\kappa)\, \bm{\iota}^\ast(g,\gamma)^{-1}\big) 
\end{flalign}
consisting of a gauge transformation $(g,\gamma)\in C^\infty(X,G)\times\overline{\Omega}^1(X,\h)$
on $X$ and a gauge transformation $(g^\diamond,\gamma^\diamond)\in 
C^\infty(M,G^{\diamond})\times \Omega^1(M,\h^{\diamond})$ on $M$.
Using the composition and inversion formulas for gauge transformations 
from Definition \ref{def:connections}, the second component of the
target of this morphism reads explicitly as
\begin{flalign}\label{eqn:edgemode gauge transformation}
(g^\diamond,\gamma^\diamond)\,(k,\kappa)\, \bm{\iota}^\ast(g,\gamma)^{-1}\,=\,
\Big(g^\diamond\,k\,\bm{\iota}^\ast(g)^{-1}\,,~
\gamma^\diamond + \alpha^{\bm{z}}_\ast(g^\diamond,\kappa) - \alpha^{\bm{z}}_\ast\big(g^\diamond\,k\, \bm{\iota}^\ast(g)^{-1},\bm{\iota}^\ast(\gamma)\big)\Big)\quad.
\end{flalign}
\end{itemize}
\end{proposition}
\begin{proof}
This follows immediately by applying the usual homotopy pullback 
construction for groupoids, see e.g.\ \cite[Appendix A]{HomEdgeModes}, to the present example.
\end{proof}

\begin{remark}\label{rem:data edge mode}
This result is analogous to the case of $4$-dimensional semi-holomorphic Chern-Simons theory
from \cite[Section 4.2]{Benini:2020skc}. The gauge transformation component
$(k,\kappa)\in C^\infty(M,G^{\bm{z}})\times\Omega^1(M,\h^{\bm{z}})$ 
of an object $\big((A,B),(k,\kappa)\big)$ in $\mathcal{F}^\diamond_{\mathrm{ho}}$
can be interpreted as an edge mode living on the defect whose role
is to witness, through its induced transformation, the boundary condition from 
Definition \ref{defi:boundarycondition} for the connection 
$(A,B)\in \overline{\Omega}^1(X,\g)\times \overline{\Omega}^2(X,\h)$.
\end{remark}

\begin{theorem}\label{theo:strict = weak}
The canonical functor
\begin{flalign}
\Phi \,:\, \mathcal{F}^\diamond\,&\longrightarrow\,\mathcal{F}^{\diamond}_\mathrm{ho}\quad,\\
\nonumber (A,B)\,&\longmapsto\,\big((A,B),(1_{G^{\bm z}},0)\big)\quad,\\
\nonumber (g,\gamma)\,&\longmapsto\,\big((g,\gamma),\bm{\iota}^\ast(g,\gamma)\big)\quad,
\end{flalign}
from the strict to the homotopy pullback is an equivalence of groupoids.
\end{theorem}
\begin{proof}
We have to verify that $\Phi$ is essentially surjective on objects and
fully faithful on morphisms.

\textit{Essential surjectivity:} Consider any object $\big((A,B),(k,\kappa)\big)$ in $\mathcal{F}^\diamond_\mathrm{ho}$.
Then the gauge transformation $(k,\kappa)\in C^\infty(M,G^{\bm{z}})\times \Omega^1(M,\h^{\bm{z}})\cong 
C^\infty(D, G)\times \Omega^1(D,\h)$ can be extended along the defect inclusion $D\subset X$, which yields
an element $(\tilde{k},\tilde{\kappa})\in C^\infty(X,G)\times \overline{\Omega}^1(X,\h)$ satisfying
$\bm{\iota}^\ast(\tilde{k},\tilde{\kappa}) = (k,\kappa)$. This defines a gauge 
transformation $(\tilde{k},\tilde{\kappa}) : (A,B)\to {}^{(\tilde{k},\tilde{\kappa})}(A,B)$ 
in $\overline{\mathrm{Con}}_{(G,H,t,\alpha)}(X)$
such that $\bm{\iota}^\ast\big({}^{(\tilde{k},\tilde{\kappa})}(A,B)\big)=
{}^{(k,\kappa)}\bm{\iota}^\ast(A,B)\in \Omega^1(M,\g^\diamond)\times \Omega^2(M,\h^\diamond)$, i.e.\
$ {}^{(\tilde{k},\tilde{\kappa})}(A,B)$ 
defines an object in the strict pullback $\mathcal{F}^\diamond$.
Essential surjectivity on objects is then proven by considering
the morphism 
\begin{flalign}
\big((\tilde{k},\tilde{\kappa}),(1_{G^\diamond},0)\big) \,:\, 
\big((A,B),(k,\kappa)\big)~\longrightarrow~\big({}^{(\tilde{k},\tilde{\kappa})}(A,B),(1_{G^{\bm{z}}},0)\big)
\end{flalign}
in $\mathcal{F}^\diamond_\mathrm{ho}$.

\textit{Fully faithfulness:} Consider two objects $(A,B)$ and $(A^\prime,B^\prime)$
in the strict pullback $\mathcal{F}^\diamond$. Then a morphism
\begin{flalign}
\big((g,\gamma),(g^\diamond,\gamma^\diamond)\big)\,:\, 
\big(((A,B),(1_{G^{\bm z}},0)\big)\,\longrightarrow\,
\big((A^\prime,B^\prime),(1_{G^{\bm z}},0)\big)
\end{flalign}
between their images in $\mathcal{F}^\diamond_{\mathrm{ho}}$ 
exists if and only if 
\begin{flalign}
(1_{G^{\bm z}},0)\, =\,(g^\diamond,\gamma^\diamond) \, (1_{G^{\bm z}},0)\, \bm{\iota}^\ast(g,\gamma)^{-1}
\,=\,  (g^\diamond,\gamma^\diamond)\,\bm{\iota}^\ast(g,\gamma)^{-1}\quad,
\end{flalign}
i.e.\ $(g^\diamond,\gamma^\diamond) = \bm{\iota}^\ast(g,\gamma) $.
This is equivalent to the statement that the morphism lies in the image of 
the functor $\Phi$, which proves fully faithfulness on morphisms.
\end{proof}

The gauge-invariant action $S_\omega: \mathcal{F}^\diamond\to\bbC$ 
from Proposition \ref{propo:action strict pullback}
can be transferred along the equivalence 
$\Phi : \mathcal{F}^\diamond\to \mathcal{F}^{\diamond}_{\mathrm{ho}}$ from Theorem \ref{theo:strict = weak}
to a gauge-invariant action $S^\mathrm{ext}_\omega :\mathcal{F}_{\mathrm{ho}}^\diamond\to\bbC $
on the groupoid $\mathcal{F}_{\mathrm{ho}}^\diamond$ in which the edge modes are manifestly included.
This equivalent point of view will be useful later to identify an integrable field theory on the
$3$-dimensional manifold $M$. The value of the extended action 
$S^\mathrm{ext}_\omega\big((A,B),(k,\kappa)\big)$ on an object in $\mathcal{F}^\diamond_\mathrm{ho}$
is determined as follows: In the proof of essential surjectivity in Theorem \ref{theo:strict = weak}, we have 
constructed a gauge transformation $(\tilde{k},\tilde{\kappa})$ from
$\big((A,B),(k,\kappa)\big)$ to the image under 
$\Phi:\mathcal{F}^\diamond\to \mathcal{F}^{\diamond}_{\mathrm{ho}}$ 
of an object ${}^{(\tilde{k},\tilde{\kappa})}(A,B)$ in $\mathcal{F}^\diamond$. We then define
\begin{flalign}
&S^\mathrm{ext}_\omega\big((A,B),(k,\kappa)\big)\,:=\,S_\omega\big({}^{(\tilde{k},\tilde{\kappa})}(A,B)\big) 
\,=\, \frac{\ii}{2 \pi} \int_X \omega\wedge \ip{\overline{F}(A) - \tfrac{1}{2} t_\ast(B)}{B}\\
\nonumber  &\quad  +\frac{1}{2} \int_M
\Big(\ipp{{}^{(k,\kappa)}\bm{\iota}^\ast(A)}{\alpha^{\bm{z}}_\ast\big(k,\bm{\iota}^\ast (B)\big)+ F_M(\kappa)}_{\omega} + \ipp{t^{\bm{z}}_\ast(\kappa)}{\dd_M \kappa +\tfrac{1}{3} [\kappa,\kappa]}_\omega\Big)\quad,
\end{flalign}
where we used Proposition \ref{propo:bdy gauge violation} and the explicit expression
\eqref{eqn:2CS 4 form} for the $2$-Chern-Simons $4$-form. Recalling also 
\eqref{eqn:gauge transformation}, which in our present context reads as
\begin{subequations}\label{eqn:diamond fields}
\begin{flalign}
{}^{(k,\kappa)}\bm{\iota}^\ast(A) \,&=\, k\, {\bm \iota}^\ast(A)\,k^{-1} -\dd_M k\, k^{-1} - t^{\bm{z}}_\ast(\kappa)\,\in\,\Omega^1(M,\g^\diamond)\quad,\\
{}^{(k,\kappa)}\bm{\iota}^\ast(B)\,&=\,\alpha^{\bm{z}}_\ast\big(k,{\bm \iota}^\ast(B)\big) -F_M(\kappa) - \alpha^{\bm{z}}_\ast\big({}^{(k,\kappa)}\bm{\iota}^\ast(A),\kappa\big)  \,\in\,\Omega^2(M,\h^\diamond)\quad,
\end{flalign}
\end{subequations}
we can rewrite the extended action further by solving the second identity
for $\alpha^{\bm{z}}_\ast\big(k,{\bm \iota}^\ast(B)\big)$.
Using also that $\ipp{{}^{(k,\kappa)}\bm{\iota}^\ast(A)}{{}^{(k,\kappa)}\bm{\iota}^\ast(B)}_\omega=0$ 
as a consequence of isotropy, we then obtain
\begin{flalign}\label{eqn:extended action simple}
&S^\mathrm{ext}_\omega\big((A,B),(k,\kappa)\big)\, =\, \frac{\ii}{2 \pi} \int_X \omega\wedge \ip{\overline{F}(A) - \tfrac{1}{2} t_\ast(B)}{B}\\
\nonumber &\quad  + \frac{1}{2}  \int_M
\Big(\ipp{{}^{(k,\kappa)}\bm{\iota}^\ast(A)}{\alpha^{\bm{z}}_\ast\big({}^{(k,\kappa)}\bm{\iota}^\ast(A),\kappa\big) + 2 F_M(\kappa)}_{\omega} + \ipp{t^{\bm{z}}_\ast(\kappa)}{\dd_M \kappa +\tfrac{1}{3} [\kappa,\kappa]}_\omega\Big)\quad. 
\end{flalign}
Note that the defect action on $M$
is a function of the edge mode $(k,\kappa)\in C^\infty(M,G^{\bm z})\times\Omega^1(M,\h^{\bm z})$ and the pullback
${\bm\iota}^\ast(A,B)\in \Omega^1(M,\g^{\bm z})\times \Omega^2(M,\h^{\bm z})$ of the connection.

\subsection{\label{subsec:EOM}Equations of motion}
We shall now derive the Euler-Lagrange equations of the extended action
\eqref{eqn:extended action simple}, which will yield bulk equations of motion
on $X$ and also defect equations of motion on the $3$-dimensional manifold $M$.
In anticipation of our applications to integrable field theories,
we shall restrict as in \cite[Section 5]{Benini:2020skc} the extended 
action to the full subgroupoid
\begin{flalign}\label{eqn:Mlegconnections}
\mathcal{F}^{\diamond,0}_{\mathrm{ho}}\,\subseteq\,\mathcal{F}^{\diamond}_{\mathrm{ho}}
\end{flalign}
whose objects $\big((A,B),(k,\kappa)\big)$ are such that the connection
$(A,B)\in \overline{\Omega}^{1,0}(X,\g)\times \overline{\Omega}^{2,0}(X,\h)$
does not have legs along $\dd \overline{z}$. (Recall from \eqref{eqn:dR quotient}
that the $\dd z$ legs are already quotiented out since they do not contribute 
to the action. Hence, connections in $\mathcal{F}^{\diamond,0}_{\mathrm{ho}}$
have only legs along the factor $M$ of the product manifold $X = M\times C$.)
Ideally, one would like to interpret this restriction as a gauge choice, but 
it is currently not clear to us if the inclusion \eqref{eqn:Mlegconnections} defines 
an equivalence of groupoids, see also \cite[Remark 5.1]{Benini:2020skc} for a similar 
issue in the $4$-dimensional case.
Restricting the extended action \eqref{eqn:extended action simple} to 
the full subgroupoid \eqref{eqn:Mlegconnections} leads to a further simplification
\begin{flalign}\label{eqn:extended action simple Mlegs}
&S^\mathrm{ext}_\omega\big((A,B),(k,\kappa)\big)\, =\, \frac{\ii}{2 \pi} \int_X \omega\wedge \ip{\overline{\partial}A}{B}\\
\nonumber &\quad +\frac{1}{2} \int_M
\Big(\ipp{{}^{(k,\kappa)}\bm{\iota}^\ast(A)}{\alpha^{\bm{z}}_\ast\big({}^{(k,\kappa)}\bm{\iota}^\ast(A),\kappa\big) + 2 F_M(\kappa)}_{\omega} + \ipp{t^{\bm{z}}_\ast(\kappa)}{\dd_M \kappa +\tfrac{1}{3} [\kappa,\kappa]}_\omega\Big)\quad,
\end{flalign}
where we recall that $\overline{\partial}$ is the Dolbeault differential on the factor $C$ 
of the product manifold $X = M\times C$.

Varying the action \eqref{eqn:extended action simple Mlegs} is slightly non-trivial because the individual components 
of an object $\big((A,B),(k,\kappa)\big)$ in $\mathcal{F}^{\diamond,0}_{\mathrm{ho}}$ 
are constrained by the condition that ${}^{(k,\kappa)}{\bm\iota}^\ast(A,B)\in 
\Omega^1(M,\g^\diamond)\times\Omega^2(M,\h^\diamond)$, see Proposition \ref{propo:homotopy pullback}.
A suitable way to parametrize such variations is by
$(A^\epsilon,B^\epsilon) := (A+\epsilon\,\mathsf{a},B+\epsilon\,\mathsf{b})$ 
and $(k^\epsilon,\kappa^\epsilon):= (e^{\epsilon \,\chi}\,k,\kappa + \epsilon\, \rho)$, 
for $\epsilon$ a small parameter.
Using \eqref{eqn:diamond fields} and performing a $1^{\mathrm{st}}$-order 
Taylor expansion in $\epsilon$, we obtain the induced variations
\begin{subequations}\label{eqn:induced variations}
\begin{flalign}
\delta \big({}^{(k,\kappa)}\bm{\iota}^\ast(A)\big)\,&=\, \big[\chi,k \,{\bm\iota}^\ast(A) \,k^{-1} -\dd_M k\,k^{-1} \big]-\dd_M\chi + k\,{\bm\iota}^\ast(\mathsf{a})\,k^{-1} - t^{\bm{z}}_\ast(\rho)\quad,\\
\delta\big({}^{(k,\kappa)}\bm{\iota}^\ast(B)\big)\,&=\,\alpha^{\bm{z}}_\ast\big(\chi,\alpha^{\bm{z}}_\ast\big(k,{\bm \iota}^\ast(B)\big)\big) + \alpha^{\bm{z}}_\ast(k,{\bm\iota}^\ast(\mathsf{b})) \\
\nonumber &\qquad - \dd_M \rho -[\kappa,\rho] -\alpha^{\bm{z}}_\ast\big(\delta \big({}^{(k,\kappa)}\bm{\iota}^\ast(A)\big),\kappa\big) 
- \alpha^{\bm{z}}_\ast\big({}^{(k,\kappa)}\bm{\iota}^\ast(A),\rho\big)\quad
\end{flalign}
\end{subequations}
for the combinations of fields entering the defect action.
Note that the constraint on the variations is then fulfilled to $1^{\mathrm{st}}$-order
in $\epsilon$ if and only if $\delta \big({}^{(k,\kappa)}\bm{\iota}^\ast(A)\big) \in \Omega^1(M,\g^\diamond)$
and $\delta \big({}^{(k,\kappa)}\bm{\iota}^\ast(B)\big) \in \Omega^2(M,\h^\diamond)$.
 
One can now work out the variation of the extended action
\eqref{eqn:extended action simple Mlegs} and one finds
after a simplification using \eqref{eqn:diamond fields} 
and \eqref{eqn:induced variations} that
\begin{flalign}\label{eqn:variation Sext}
&\delta S^{\mathrm{ext}}_\omega \,=\, \frac{\ii}{2 \pi} \int_X\omega\wedge \Big(\ip{\mathsf{a}}{\overline{\partial} B} 
+ \ip{\overline{\partial}A}{\mathsf{b}}\Big) \\
\nonumber  &\quad + \int_M \Big(\ipp{k^{-1}\chi k}{{\bm\iota}^\ast\big(\dd_M B 
+ \alpha_\ast(A,B)\big)}_\omega + \ipp{{\bm\iota}^\ast\big(F_M(A) - t_\ast(B)\big)}{\alpha^{\bm{z}}_\ast(k^{-1},\rho)}_\omega\Big)\quad,
\end{flalign}
where we use that under the pullback $\bm{\iota}^\ast$ only the $\dd_M$ component 
of the differential $\overline{\dd} =\dd_M + \overline{\partial}$ 
on the quotient de Rham calculus $\overline{\Omega}^\bullet(X)$ from \eqref{eqn:dR quotient} survives.
Let us also note that the term 
$-\ipp{\delta \big({}^{(k,\kappa)}\bm{\iota}^\ast(A)\big)}{{}^{(k,\kappa)}\bm{\iota}^\ast(B)}_{\omega}$ 
on $M$ which one finds in this calculation vanishes manifestly as a consequence of isotropy. 
We summarize this result in the following
\begin{proposition}\label{propo:EOM}
The Euler-Lagrange equations of the extended action $S_\omega^{\mathrm{ext}}$ in \eqref{eqn:extended action simple Mlegs}
on the full subgroupoid \eqref{eqn:Mlegconnections} are given by the bulk equations of motion on $X$ 
\begin{flalign}\label{eqn:bulk EOM simple}
\omega\wedge \overline{\partial} B \,=\,0\quad,\qquad 
\omega\wedge \overline{\partial}A\,=\,0\quad,
\end{flalign}
and by the defect equations of motion on $M$
\begin{flalign}\label{eqn:defect EOM simple}
{\bm \iota}^\ast\big(\dd_M B + \alpha_\ast(A,B)\big)\,=\ 0\quad,\qquad
{\bm \iota}^\ast\big(F_M(A) - t_\ast(B)\big)\,=\ 0\quad.
\end{flalign}
\end{proposition}
\begin{proof}
To derive the bulk equations, one uses \eqref{eqn:variation Sext} for all variations 
$(A^\epsilon,B^\epsilon) := (A+\epsilon\,\mathsf{a},B+\epsilon\,\mathsf{b})$ 
and $(k^\epsilon,\kappa^\epsilon):= (k,\kappa)$ with $\mathsf{a}\in\overline{\Omega}^{1,0}(X,\g)$ 
and $\mathsf{b}\in\overline{\Omega}^{2,0}(X,\h)$
supported on the complement $X\setminus D$ of the defect $D\subset X$.
Note that such variations manifestly satisfy the constraints 
$\delta \big({}^{(k,\kappa)}\bm{\iota}^\ast(A)\big) \in \Omega^1(M,\g^\diamond)$ and 
$\delta\big({}^{(k,\kappa)}\bm{\iota}^\ast(B)\big)\in \Omega^2(M,\h^\diamond)$
on the induced variations.

To derive the defect equations, we consider any 
variation $(k^\epsilon,\kappa^\epsilon):= (e^{\epsilon\,\chi}\,k,\kappa+\epsilon\,\rho)$
of the edge modes with $\chi\in C^\infty(M,\g^{{\bm z}})$ and $\rho\in\Omega^1(M,\h^{\bm z})$.
From the explicit form of the induced variations \eqref{eqn:induced variations},
one finds that the conditions $\delta \big({}^{(k,\kappa)}\bm{\iota}^\ast(A)\big) = 0$ and 
$\delta \big({}^{(k,\kappa)}\bm{\iota}^\ast(B)\big) =0$,
which in particular imply the constraint, can be solved uniquely for 
${\bm \iota}^\ast (\mathsf{a})\in\Omega^1(M,\g^{\bm z})\cong\Omega^1(D,\g)$
and ${\bm \iota}^\ast(\mathsf{b})\in \Omega^2(M,\h^{\bm z})\cong\Omega^2(D,\h)$. 
Choosing any extensions $\mathsf{a}\in \overline{\Omega}^{1,0}(X,\g)$ and
$\mathsf{b}\in\overline{\Omega}^{2,0}(X,\h)$ of these forms 
along the defect inclusion $D\subset X$ gives
a compatible variation of the connection 
$(A^\epsilon,B^\epsilon):= (A + \epsilon\,\mathsf{a},B+
\epsilon \,\mathsf{b})$. The result then follows from \eqref{eqn:variation Sext}.
\end{proof}


\section{\label{sec:higherpoles}$5d$ $2$-Chern-Simons theory with higher poles}
All constructions and results from Section \ref{sec:simplepoles} 
can be generalized to the case where the meromorphic $1$-form $\omega$ on $\CP$
has higher-order poles by using the concept of \textit{regularized integrals} from 
\cite{Li:2020ljm} and \cite{Benini:2020skc}. We shall now briefly state the relevant results
and refer the reader to \cite[Section 3]{Benini:2020skc} for more details and complete proofs
in the similar case of $4$-dimensional semi-holomorphic Chern-Simons theory.

As in the previous section, we consider the $5$-dimensional manifold $X = M\times C$
with $M = \bbR^3$ the $3$-dimensional Cartesian space, interpreted as spacetime, 
and $C =\CP\setminus {\bm \zeta} $ the Riemann sphere with all zeros ${\bm \zeta}\subset \CP$ of $\omega$ removed.
We can and will choose a global coordinate $z : C\to\bbC$ on $C$
since by our hypotheses $\vert {\bm \zeta}\vert \geq 1$. Using this coordinate,
we can write the meromorphic $1$-form as
\begin{flalign}\label{eqn:omega explicit}
\omega\,=\, \sum_{x\in{\bm z}} \sum_{p=0}^{n_x-1}\frac{k^x_p}{(z-x)^{p+1}}\,\dd z\quad,
\end{flalign}
where $n_x\in\mathbb{Z}_{\geq 1}$ denotes the order of the pole $x\in{\bm z}$ and $k^x_p\in\bbC$ are constants.

Let us denote by $n\coloneqq \mathrm{max}(n_x)_{x\in{\bm z}}$ the maximal order among all poles of $\omega$.
We introduce the corresponding Weil algebra
\begin{flalign}
\mathcal{T}^n\,\coloneqq\, \bbC[\epsilon]\big/ (\epsilon^n)
\end{flalign}
of order $n$ and define the holomorphic $(n-1)$-jet prolongation 
\begin{flalign}
j_X^\ast\,:\,\overline{\Omega}^q(X)~\longrightarrow\,\overline{\Omega}^q(X)\otimes\mathcal{T}^n~~,\quad\eta\,\longmapsto\,
\sum_{p=0}^{n-1}\frac{1}{p!}\,\partial_z^p\eta\otimes \epsilon^p\quad,
\end{flalign}
where we recall that $(\overline{\Omega}^\bullet(X),\overline{\dd})$
denotes the quotient de Rham calculus from \eqref{eqn:dR quotient}.
The generalization of the action \eqref{eqn:5d CS action} to the case of higher-order poles
is then given by the regularized integral
\begin{subequations}\label{eqn:5d CS action higher}
\begin{flalign}
S_\omega(A,B)\,\coloneqq\, \frac{\ii}{2 \pi} \,\,\dashint_X \omega\wedge \CS(A,B)\,\coloneqq\,
\frac{\ii}{2 \pi} \int_X\big(\omega\wedge j^\ast_X\CS(A,B)\big)_{\mathrm{reg}}\quad,
\end{flalign}
where 
\begin{flalign}
\big(\omega\wedge j^\ast_X\CS(A,B)\big)_{\mathrm{reg}} \,\coloneqq\,
\sum_{x\in{\bm z}}\sum_{p=0}^{n_x-1}\frac{k_p^x}{z-x}\,\dd z\wedge \frac{1}{p!}\partial_z^p\CS(A,B)\quad.
\end{flalign}
\end{subequations}
The key property of this regularization construction (see \cite[Lemma 3.2]{Benini:2020skc})
is that the $5$-form $\big(\omega\wedge j^\ast_X\CS(A,B)\big)_{\mathrm{reg}}$
is locally integrable near all poles $x\in{\bm z}$ of $\omega$ 
and that it agrees with the ordinary wedge product 
$ \omega\wedge \CS(A,B) = \big(\omega\wedge j^\ast_X\CS(A,B)\big)_{\mathrm{reg}} + \dd \psi$ 
up to an exact term which is singular and non-integrable near the poles of $\omega$.
Note that the action coincides with \eqref{eqn:5d CS action} in the case where $\omega$
has only simple poles, i.e.\ $n_x=1$ for all $x\in{\bm z}$.

The generalization of the defect \eqref{eqn:surface defect}
to the higher-order pole case is given by the formal manifold
\begin{flalign}\label{eqn:surface defect higher}
\widehat{D}\,\coloneqq\,\bigsqcup_{x\in \bm{z}}\big( M\times \ell \mathcal{T}_x^{n_x}\big)\quad,
\end{flalign}
where $\ell \mathcal{T}_x^{n_x}$ denotes the locus (in the sense of synthetic geometry \cite{Kock})
of the Weil algebra $\mathcal{T}_x^{n_x}=\bbC[\epsilon_x]/(\epsilon_x^{n_x})$ 
of order $n_x$ given by the order of the pole $x\in{\bm z}$. 
One should interpret $\ell \mathcal{T}_x^{n_x}$ as an infinitesimally thickened point.
The formal defect \eqref{eqn:surface defect higher} embeds
\begin{flalign}
{\bm j}\,:\, \widehat{D}\,\longhookrightarrow \,X
\end{flalign}
into $X$, which induces pullback maps
\begin{subequations}\label{eqn:jpull maps}
\begin{flalign}
{\bm j}^\ast\,:\, C^\infty(X,N)~&\longrightarrow\, C^\infty(\widehat{D},N)\,\cong\,C^\infty(M,N^{\hat{\bm z}})\\
{\bm j}^\ast\,:\, \overline{\Omega}^q(X,V)~&\longrightarrow\, \Omega^q(\widehat{D},V)\,\cong\,\Omega^q(M,V^{\hat{\bm z}})
\end{flalign}
\end{subequations}
generalizing \eqref{eqn:iota pullback maps}, where
\begin{flalign}\label{eqn:jet manifold}
N^{\hat{{\bm z}}}\,\coloneqq\, \prod_{x\in{\bm z}} C^\infty(\ell \mathcal{T}_x^{n_x},N)\quad
\end{flalign}
denotes the product of $(n_x-1)$-jet manifolds over the manifold $N$ and
\begin{flalign}
V^{\hat{\bm z}}\,\coloneqq\,\prod_{x\in{\bm z}}\big( V \otimes \mathcal{T}_x^{n_x}\big)\quad.
\end{flalign}
Explicitly, the pullback maps \eqref{eqn:jpull maps}
are pullbacks along the inclusions $\iota_x : M\times\{x\} \hookrightarrow X$ 
of holomorphic jet prolongations, i.e.\
\begin{flalign}\label{eqn:jpull maps explicit}
\bm{j}^\ast(\,\cdot\,)\,=\,\bigg(\sum_{p=0}^{n_x-1}\tfrac{1}{p!} \iota_x^\ast\big(\partial_z^p (\,\cdot\,)\big)\otimes \epsilon_x^p\bigg)_{x\in{\bm z}}\quad.
\end{flalign}
The generalization of the pairing \eqref{eqn:defectpairing simple} 
to the higher-order pole case is given by
\begin{flalign}\label{eqn:defectpairing higher}
\ipp{\,\cdot\,}{\,\cdot\,}_{\omega}\,:\, \g^{\hat{\bm{z}}}\otimes\h^{\hat{\bm{z}}}\,\longrightarrow\bbC~~,\quad
\mathcal{X}\otimes \mathcal{Y}\,\longmapsto\,\ipp{\mathcal{X}}{\mathcal{Y}}_\omega\,\coloneqq\,
\sum_{x\in\bm{z}}\sum_{p,q=0}^{n_x-1} k_{p+q}^x\,\ip{\mathcal{X}^x_p}{\mathcal{Y}^x_q}\quad,
\end{flalign}
where $\mathcal{X} = 
\big(\sum_{p=0}^{n_x-1} \mathcal{X}^x_p\otimes \epsilon_x^p\big)_{x\in{\bm z}}\in \g^{\hat{\bm{z}}}$,
$\mathcal{Y} = \big(\sum_{q=0}^{n_x-1} \mathcal{Y}^x_q\otimes \epsilon_x^q\big)_{x\in{\bm z}}\in \h^{\hat{\bm{z}}}$
and the coefficients $k_{p+q}^x$ are determined from $\omega$, see \eqref{eqn:omega explicit}.
(Our convention is that $k_{p+q}^x =0$ for all $p+q>n_x-1$.) 

The following result is the generalization of Proposition \ref{propo:bdy gauge violation}
to the case of higher-order poles.
\begin{proposition}\label{propo:bdy gauge violation higher}
Under a gauge transformation $(g,\gamma) : (A,B) \to {}^{(g,\gamma)}(A,B)$, 
with $g\in C^\infty(X, G)$ and $\gamma\in\overline{\Omega}^1(X,\h)$, the 
regularized action \eqref{eqn:5d CS action higher} transforms as
\begin{flalign}
\nonumber S_\omega\big({}^{(g,\gamma)}(A,B)\big) &\,=\,  S_{\omega}(A,B) + \frac{1}{2} \int_M
\Big(\ipp{\bm{j}^\ast(g)\, \bm{j}^\ast(A)\, \bm{j}^\ast(g)^{-1}}{F_M\big({\bm{j}^\ast (\gamma)}\big)}_{\omega}\\ 
\nonumber &\quad + \ipp{\bm{j}^\ast\big(t_\ast(\gamma)\big)}{\dd_M \bm{j}^\ast (\gamma) +\tfrac{1}{3}\big[\bm{j}^\ast (\gamma),\bm{j}^\ast (\gamma)\big]}_\omega\\
&\quad - \ipp{\dd_M \bm{j}^\ast (g) \,\bm{j}^\ast (g)^{-1}+\bm{j}^\ast\big(t_\ast(\gamma)\big)}{ \bm{j}^\ast\big(\alpha_\ast(g,B)\big)+F_M\big({\bm{j}^\ast (\gamma)}\big)}_\omega\Big)\quad.
\end{flalign}
\end{proposition}
\begin{proof}
This is a direct consequence of Proposition \ref{propo:2CS 4form gauge}
and the same arguments as in \cite[Lemma 3.3 and Proposition 3.4]{Benini:2020skc}.
\end{proof}

To impose boundary conditions, we choose an isotropic crossed submodule
\begin{flalign}
(G^\diamond,H^\diamond,t^{\hat{\bm z}},\alpha^{\hat{\bm z}})\,\subseteq\, (G^{\hat{\bm z}},H^{\hat{\bm z}},t^{\hat{\bm z}},\alpha^{\hat{\bm z}})
\end{flalign}
with respect to the pairing \eqref{eqn:defectpairing higher},
where as a consequence of \eqref{eqn:jet manifold} the ambient crossed module consists of products of jet groups
\begin{flalign}\label{eqn:jetgroups}
G^{\hat{{\bm z}}}\,=\, \prod_{x\in{\bm z}} C^\infty(\ell \mathcal{T}_x^{n_x},G)\quad,\qquad
H^{\hat{{\bm z}}}\,=\, \prod_{x\in{\bm z}} C^\infty(\ell \mathcal{T}_x^{n_x},H)\quad,
\end{flalign}
see also \cite{Vizman} for a more explicit description of such jet groups.
The construction of the groupoid $\mathcal{F}^\diamond_{\mathrm{ho}}$ of 
boundary conditioned fields with edge modes from Proposition \ref{propo:homotopy pullback} generalizes in 
the evident way: One simply replaces the crossed module
$(G^{{\bm z}},H^{{\bm z}},t^{\bm{z}},\alpha^{\bm{z}})$ of product groups
by the crossed module $(G^{\hat{\bm z}},H^{\hat{\bm z}},t^{\hat{\bm z}},\alpha^{\hat{\bm z}})$
of products of jet groups, and further replaces 
the pullback maps ${\bm \iota}^\ast$ by the maps
${\bm j}^\ast$ from \eqref{eqn:jpull maps}.
Following the same steps as in Section \ref{sec:simplepoles}, one then 
arrives at the extended action
\begin{flalign}\label{eqn:extended action higher}
&S^\mathrm{ext}_\omega\big((A,B),(k,\kappa)\big)\, =\, \frac{\ii}{2 \pi}\,\, \dashint_X \omega\wedge \ip{\overline{F}(A) - \tfrac{1}{2} t_\ast(B)}{B}\\
\nonumber  &\quad +\frac{1}{2} \int_M
\Big(\ipp{{}^{(k,\kappa)}\bm{j}^\ast(A)}{\alpha^{\hat{\bm z}}_\ast\big({}^{(k,\kappa)}\bm{j}^\ast(A),\kappa\big) + 2 F_M(\kappa)}_{\omega} + \ipp{t^{\hat{\bm z}}_\ast(\kappa)}{\dd_M \kappa +\tfrac{1}{3} [\kappa,\kappa]}_\omega\Big)\quad,
\end{flalign}
which generalizes \eqref{eqn:extended action simple} to the higher-order pole case.
Restricting as in \eqref{eqn:Mlegconnections} to the full subgroupoid 
\begin{flalign}\label{eqn:Mlegconnectionshigher}
\mathcal{F}^{\diamond,0}_{\mathrm{ho}}\,\subseteq\, \mathcal{F}^\diamond_{\mathrm{ho}}
\end{flalign}
whose objects $\big((A,B),(k,\kappa)\big)$ are such that the connection
$(A,B)\in \overline{\Omega}^{1,0}(X,\g)\times \overline{\Omega}^{2,0}(X,\h)$
does not have legs along $\dd \overline{z}$,
we obtain as in \eqref{eqn:extended action simple Mlegs} a further simplification
\begin{flalign}\label{eqn:extended action higher Mlegs}
 &S^\mathrm{ext}_\omega\big((A,B),(k,\kappa)\big)\, =\, \frac{\ii}{2 \pi} \,\, \dashint_X \omega\wedge \ip{\overline{\partial}A}{B}\\
\nonumber &\quad  + \frac{1}{2} \int_M
\Big(\ipp{{}^{(k,\kappa)}\bm{j}^\ast(A)}{\alpha^{\hat{\bm z}}_\ast\big({}^{(k,\kappa)}\bm{j}^\ast(A),\kappa\big) + 2 F_M(\kappa)}_{\omega} + \ipp{t^{\hat{\bm z}}_\ast(\kappa)}{\dd_M \kappa +\tfrac{1}{3} [\kappa,\kappa]}_\omega\Big)\quad.
\end{flalign}
In complete analogy to Proposition \ref{propo:EOM}, one can work out the
variation of this action, which 
yields the bulk equations of motion on $X$ 
\begin{flalign}\label{eqn:bulk EOM higher}
\omega\wedge \overline{\partial} B \,=\,0\quad,\qquad
\omega\wedge \overline{\partial} A\,=\,0\quad,
\end{flalign}
and the defect equations of motion on $M$
\begin{flalign}\label{eqn:defect EOM higher}
{\bm j}^\ast\big(\dd_M B + \alpha_\ast(A,B)\big)\,=\ 0\quad,\qquad
{\bm j}^\ast\big(F_M(A) - t_\ast(B)\big)\,=\ 0\quad,
\end{flalign}
where we again use that under the pullback $\bm{j}^\ast$ only the $\dd_M$ component 
of the differential $\overline{\dd} =\dd_M + \overline{\partial}$ survives.


\section{\label{sec:integrable}Construction of $3d$ integrable field theories}
With our preparations from Sections \ref{sec:simplepoles} and \ref{sec:higherpoles},
we are now ready to construct $3$-dimensional integrable field theories on $M$.
The key observation which makes this endeavor possible is that the bulk equations of motion
\eqref{eqn:bulk EOM higher} imply that the connection 
$(A,B)\in \overline{\Omega}^{1,0}(X,\g)\times \overline{\Omega}^{2,0}(X,\h)$
is holomorphic away from the zeros $\bm{\zeta}\subset \CP$ of the meromorphic
$1$-form $\omega$ and that the defect equations of motion
\eqref{eqn:defect EOM higher} are the pullback to the defect of the $M$-relative 
flatness condition for the connection $(A,B)$. These are almost the properties which one requires for
a higher Lax connection. However, as in the context of $4$-dimensional semi-holomorphic Chern-Simons theory, see
\cite[Section 5]{Benini:2020skc}, the following crucial points need further attention:
\begin{enumerate}
\item To qualify as a Lax connection, the connection
$(A,B)\in \overline{\Omega}^{1,0}(X,\g)\times \overline{\Omega}^{2,0}(X,\h)$ must 
not only be holomorphic away from the zeros of $\omega$, which is
implied by the bulk equations of motion \eqref{eqn:bulk EOM higher}, but it 
further must be meromorphic on all of $\CP$.

\item The flatness conditions implied by the defect equations 
of motion \eqref{eqn:defect EOM higher} must lift along $\bm{j}^\ast$
to the $M$-relative flatness conditions $\dd_M B + \alpha_\ast(A,B) =0$
and $F_M(A) - t_\ast(B) = \dd_M A+\tfrac{1}{2}[A,A] - t_\ast(B)=0$ on $X$.

\item The boundary conditions ${}^{(k,\kappa)}\bm{j}^\ast(A,B) \in \Omega^1(M,\g^\diamond)
\times \Omega^2(M,\h^\diamond)$ should admit a unique solution for
the Lax connection $(A,B)\in \overline{\Omega}^{1,0}(X,\g)\times \overline{\Omega}^{2,0}(X,\h)$
in terms of the edge mode fields $(k,\kappa)\in C^\infty(M,G^{\hat{\bm{z}}})\times \Omega^1(M,\h^{\hat{\bm{z}}})$,
making the latter the only degrees of freedom of the $3$-dimensional integrable field theory on $M$.
\end{enumerate}

Although the uniqueness assumption in the third point is not strictly necessary,
it holds in the vast majority of examples in the context of $4$-dimensional semi-holomorphic Chern-Simons theory,
see for instance \cite{Lacroix:2020flf}. We will thus focus in the present work on the situation where the boundary
conditions admit a unique solution and come back to this in Section \ref{subsec:degreecounting} below.

The first two issues can be successfully solved by 
considering a special class of solutions $(A,B)\in \overline{\Omega}^{1,0}(X,\g)\times \overline{\Omega}^{2,0}(X,\h)$
of the bulk equations of motion \eqref{eqn:bulk EOM higher} 
which have a specific behavior towards the zeros of $\omega$.
The following definition originated in \cite[Section 5]{Benini:2020skc}.
\begin{definition}\label{def:admissible}
\begin{itemize}
\item[(a)] Let $V$ be a vector space and consider
the vector space $\overline{\Omega}^{q,0}(X,V)$ of $V$-valued $(q,0)$-forms on the product
manifold $X=M\times C$. We define $\overline{\Omega}^{q,0}_{\omega}(X,V)\subseteq \overline{\Omega}^{q,0}(X,V)$
as the subspace consisting of all $V$-valued $(q,0)$-forms which are meromorphic
on $\CP$ with poles at each zero $y\in\bm{\zeta}$ of $\omega$ 
of order at most that of the zero.

\item[(b)] A connection $(A,B)\in \overline{\Omega}_{\omega}^{1,0}(X,\g)\times 
\overline{\Omega}_{\omega}^{2,0}(X,\h)$ is called \textit{admissible} 
if its $M$-relative curvature
\begin{flalign}
\mathsf{curv}_M(A,B)\,:=\,\Big(F_M(A) - t_\ast(B),\, \dd_M B + \alpha_\ast(A,B)\Big)\,\in\,
\overline{\Omega}^{2,0}_{\omega}(X,\g)\times \overline{\Omega}^{3,0}_{\omega}(X,\h)
\end{flalign}
lies in the subspaces from item (a).
\end{itemize}
\end{definition}

The following key result has been proven in \cite[Lemma 5.5 and Proposition 5.6]{Benini:2020skc}.
\begin{proposition}\label{prop:admissible}
Let $\big((A,B),(k,\kappa)\big)\in \mathcal{F}^{\diamond,0}_{\mathrm{ho}}$ be
any object in the groupoid \eqref{eqn:Mlegconnectionshigher} such 
that $(A,B)\in \overline{\Omega}_{\omega}^{1,0}(X,\g)\times 
\overline{\Omega}_{\omega}^{2,0}(X,\h)$ is an admissible connection.
Then $\big((A,B),(k,\kappa)\big)$ solves the bulk equations of motion \eqref{eqn:bulk EOM higher} 
and the defect equations of motion \eqref{eqn:defect EOM higher} are equivalent
to the $M$-relative flatness conditions
\begin{flalign}
\dd_M B + \alpha_\ast(A,B)\,=\,0 \quad,\qquad F_M(A) - t_\ast(B)\,=\,0
\end{flalign}
on the $5$-dimensional manifold $X$.
\end{proposition}

\subsection{\label{subsec:degreecounting} A notion of maximality for isotropic crossed submodules}
In the setting of $4$-dimensional semi-holomorphic Chern-Simons theory, a powerful approach
to solving the boundary conditions for the gauge field in terms of the edge modes was
developed in \cite{Lacroix:2020flf}. An important requirement in the construction of \cite{Lacroix:2020flf}
is that the isotropic subalgebra entering the boundary conditions is \emph{maximal},
namely that its dimension should be half that of the defect Lie algebra.
A sufficient additional condition was then identified for the existence of a \emph{unique} solution to these
boundary conditions, see \cite[Lemma 2.2 and Section 4.4]{Lacroix:2020flf}.

Of course, it would be very desirable to generalize the full construction of \cite{Lacroix:2020flf} to the present
$3$-dimensional setting. Our more modest goal in this section is to identify a suitable notion
of \emph{maximality} for the isotropic crossed submodule
$(G^\diamond,H^\diamond,t^{\hat{\bm z}},\alpha^{\hat{\bm z}})\subseteq
(G^{\hat{\bm{z}}},H^{\hat{\bm{z}}},t^{\hat{\bm z}},\alpha^{\hat{\bm z}})$, for which it can be expected that
the boundary conditions admit a unique solution for the Lax connection
$(A,B)\in \overline{\Omega}^{1,0}(X,\g)\times \overline{\Omega}^{2,0}(X,\h)$
in terms of the edge modes $(k,\kappa)\in C^\infty(M,G^{\hat{\bm{z}}})\times \Omega^1(M,\h^{\hat{\bm{z}}})$.
Specifically, guided by the $4$-dimensional semi-holomorphic Chern-Simons case as an analogy, we
will make a particular choice of ansatz for the meromorphic connection $(A,B)$ and give a simple counting
argument to determine the expected dimensions of $\g^\diamond$ and $\h^\diamond$ for the unique solvability
of the boundary conditions. In the next subsections we will then present concrete examples
which fit within this proposed scheme.
For our argument we assume the following specific properties
of the meromorphic $1$-form $\omega\in\Omega^1(\CP)$.
\begin{assumption}\label{assu:divisible}
The meromorphic $1$-form $\omega\in\Omega^1(\CP)$ 
from \eqref{eqn:omega explicit} only has simple zeros and
its total number of poles (counting multiplicities) is divisible by $3$, i.e.\
\begin{flalign}
\vert \bm{z}\vert \,=\,\sum_{x\in\bm{z}}n_x\,\in\, 3\,\mathbb{Z}_{\geq 1}\quad.
\end{flalign}
\end{assumption}

Let us pick any zero $y_0\in\bm{\zeta}$ of $\omega$ and denote by
$\bm{\zeta}^\prime := \bm{\zeta}\setminus \{y_0\}$ the set of the remaining zeros.
As a consequence of Assumption \ref{assu:divisible}
and the identity $\vert\bm{z}\vert = \vert\bm{\zeta}\vert+2$
for the total numbers of poles and zeros of any meromorphic $1$-form on the Riemann sphere $\CP$,
we observe that $\vert \bm{\zeta}^\prime\vert \in 3\,\mathbb{Z}_{\geq 0}$
is either $0$ (in the case where $\omega$ has only one zero) 
or divisible by $3$. This allows us to choose a decomposition
\begin{flalign}
\bm{\zeta}^\prime\,=\,\bm{\zeta}_1 \sqcup \bm{\zeta}_2\sqcup \bm{\zeta}_3
\end{flalign}
into three subsets of the same size $\vert\bm{\zeta}_1\vert = \vert\bm{\zeta}_2\vert =\vert\bm{\zeta}_3\vert $.
The labels $1,2,3$ will correspond to a choice of coordinates $u_i$, for $i=1,2,3$,
on the $3$-dimensional spacetime  $M=\mathbb{R}^3$. 

Using the above choices, we consider the following ansatz 
\begin{subequations}\label{eqn:admissibleansatz}
\begin{flalign}
A \,&=\, \sum_{i=1}^3 \bigg(A_c^i + \sum_{y\in \bm{\zeta}_i}\frac{A_y^i}{z-y}\bigg)\,\dd u_i\,\in\,
\overline{\Omega}_{\omega}^{1,0}(X,\g)\quad,\\
B \,&=\, \sum_{i,j=1}^3\bigg(B_c^{ij} + \sum_{y\in \bm{\zeta}_i\sqcup \bm{\zeta}_j\sqcup\{y_0\}}\frac{B_y^{ij}}{z-y}\bigg)\,\dd u_i\wedge \dd u_j\,\in\, \overline{\Omega}_{\omega}^{2,0}(X,\h)\
\end{flalign}
\end{subequations}
for a connection $(A,B)\in \overline{\Omega}_{\omega}^{1,0}(X,\g)\times 
\overline{\Omega}_{\omega}^{2,0}(X,\h)$,
where $A^i_c,A^i_y \in C^\infty(M,\g)$ and $B^{ij}_c,B^{ij}_y\in C^\infty(M,\h)$ 
are arbitrary coefficient functions depending only on $M$.
Note that the connection \eqref{eqn:admissibleansatz} 
is admissible in the sense of Definition \ref{def:admissible}.
The number of independent degrees of freedom of this connection, which we
count as scalars on $M$, is given by
\begin{subequations}\label{eqn:DOFcount}
\begin{flalign}
\mathrm{dof}(A) \,&=\, \sum_{i=1}^3\big(\vert \bm{\zeta}_i\vert +1\big)\times \dim(\g)\,=\,\big(\vert\bm{\zeta}\vert + 2\big)\times\dim(\g)\quad,\\
\mathrm{dof}(B) \,&=\,\sum_{i>j}\big(\vert \bm{\zeta}_i \vert+ \vert \bm{\zeta}_j\vert + 2\big)\times\dim(\h)\,=\,
2\times \big(\vert\bm{\zeta}\vert+2\big)\times\dim(\h)\quad.
\end{flalign}
\end{subequations}
These degrees of freedom are constrained by the boundary conditions
\begin{subequations}\label{eqn:bdyconditionsj}
\begin{flalign}
{}^{(k,\kappa)}\bm{j}^\ast(A) \,&=\, k\, {\bm j}^\ast(A)\,k^{-1} -\dd_M k\, k^{-1} - t^{\hat{\bm z}}_\ast(\kappa)\,\in\,\Omega^1(M,\g^\diamond)\quad,\\
{}^{(k,\kappa)}\bm{j}^\ast(B)\,&=\,\alpha^{\hat{\bm z}}_\ast\big(k,{\bm j}^\ast(B)\big) -F_M(\kappa) - 
\alpha^{\hat{\bm z}}_\ast\big({}^{(k,\kappa)}\bm{j}^\ast(A),\kappa\big)  \,\in\,\Omega^2(M,\h^\diamond)\quad.
\end{flalign}
\end{subequations}
Counting the number of boundary conditions (again as scalars on $M$), one finds
\begin{subequations}\label{eqn:bdycount}
\begin{flalign}
\mathrm{bdy}(A) \,&=\, 3\times \big(\dim(\g^{\hat{\bm z}}) - \dim(\g^\diamond)\big)\,=\, 3\times \big(\vert\bm{z}\vert \times \dim(\g) - \dim(\g^\diamond)\big) \quad,\\
\mathrm{bdy}(B) \,&=\, 3\times \big(\dim(\h^{\hat{\bm z}}) - \dim(\h^\diamond)\big)\,=\, 3\times \big(\vert \bm{z}\vert\times \dim(\h) - \dim(\h^\diamond)\big)\quad.
\end{flalign}
\end{subequations}
For the unique solvability for $(A,B)$ of the boundary conditions \eqref{eqn:bdyconditionsj} 
one requires that there are as many boundary conditions as there are degrees of freedom, i.e.\
$\mathrm{dof}(A) = \mathrm{bdy}(A) $ and $\mathrm{dof}(B) = \mathrm{bdy}(B)$.
Using also the identity $\vert\bm{z}\vert = \vert\bm{\zeta}\vert +2$ we conclude that a necessary condition
for the unique solvability of \eqref{eqn:bdyconditionsj}, assuming the ansatz \eqref{eqn:admissibleansatz}, is given by
\begin{flalign}\label{eqn:solve conditions}
\dim(\g^\diamond) \,=\,\tfrac{2}{3} \dim(\g^{\hat{\bm{z}}}) \quad,\qquad
\dim(\h^\diamond) \,=\,\tfrac{1}{3} \dim(\h^{\hat{\bm{z}}}) \,=\,\tfrac{1}{3} \dim(\g^{\hat{\bm{z}}})\quad,
\end{flalign}
where in the last step we used that $\dim(\h) = \dim(\g)$ as a consequence of the non-degenerate
pairing \eqref{eqn: pairing}. The condition \eqref{eqn:solve conditions} implies in particular
that the isotropic crossed submodule
$(G^\diamond,H^\diamond,t^{\hat{\bm z}},\alpha^{\hat{\bm z}})\subseteq 
(G^{\hat{\bm{z}}},H^{\hat{\bm{z}}},t^{\hat{\bm z}},\alpha^{\hat{\bm z}})$
must be \textit{maximal} in the sense that its total dimension
$\dim(\g^\diamond) + \dim(\h^\diamond) = \dim(\g^{\hat{\bm{z}}})$ is half
of the total dimension $\dim(\g^{\hat{\bm{z}}}) + \dim(\h^{\hat{\bm{z}}}) = 2\,\dim(\g^{\hat{\bm{z}}})$
of the ambient crossed module.

\subsection{\label{subsec:ChernSimons}Toy-example: $3$-dimensional Chern-Simons theory}
In this subsection we show how one can recover the usual $3$-dimensional Chern-Simons
theory as a defect theory of our $5$-dimensional semi-holomorphic $2$-Chern-Simons theory.
For this we consider the crossed module of Lie groups
$(G,G,\id,\Ad)$ with $t =\id : G\to G\,,~g\mapsto g$ the identity map
and $\alpha = \Ad : G\times G\to G\,,~(g,g^\prime)\mapsto g\,g^\prime\,g^{-1}$ the adjoint action.
The associated crossed module of Lie algebras is given by $(\g,\g,\id,\ad)$
with $\ad: \g\otimes\g \to \g\,,~(x,x^\prime)\mapsto [x,x^\prime]$ the Lie algebra adjoint action.
For the non-degenerate pairing in \eqref{eqn: pairing} 
we take any non-degenerate $\Ad$-invariant symmetric pairing 
$\ip{\,\cdot\,}{\,\cdot\,} :\g\otimes \g\to \mathbb{C} $ on the Lie algebra $\g$.
For the meromorphic $1$-form $\omega\in\Omega^1(\CP)$ we choose
\begin{flalign}
\omega\,=\, \frac{1-z}{z}\,\dd z\,=\, \frac{\dd z}{z} - \dd z\quad,
\end{flalign}
which has a simple zero at $z=1$,
a simple pole at $z=0$ and a double pole at $z=\infty$. 
Note that Assumption \ref{assu:divisible} is satisfied.
To avoid confusion, let us highlight that we choose for convenience in this and 
the next example a coordinate $z$ on $\CP$ in which $\infty$ is a pole of $\omega$,
while in Sections \ref{sec:simplepoles} and \ref{sec:higherpoles}
the coordinate was chosen such that $\infty$ corresponds to a zero of $\omega$.

The associated crossed module of jet groups \eqref{eqn:jetgroups}
is given by $(G^{\hat{\bm{z}}},G^{\hat{\bm{z}}},\id,\Ad)$
with
\begin{subequations}
\begin{flalign}
G^{\hat{\bm{z}}}\,=\, G\times \big(G\ltimes \tilde{\g}\big)\quad,
\end{flalign}
where the factor $G$ corresponds to the simple pole at $z=0$ and the
semi-direct product $G\ltimes \tilde{\g}$ corresponds to the double pole at $z=\infty$.
The notation $\tilde{\g}$ is used to distinguish between the Lie algebra $\g$ and the Abelian
Lie group $\tilde{\g}:=\g$ with group operation $+$ and identity element $0\in\g$.
The group structure reads explicitly as
\begin{flalign}
\big(g_0,(g_\infty,x_\infty)\big)\,\big(g_0^\prime,(g_\infty^\prime,x_\infty^\prime)\big) 
\,=\, \big(g_0\,g_0^\prime, 
\big(g_\infty\,g_\infty^\prime,\, x_\infty + g_\infty\,x_\infty^\prime\,g_\infty^{-1}\big)\big)\quad,
\end{flalign}
\end{subequations}
for all $\big(g_0,(g_\infty,x_\infty)\big),\big(g^\prime_0,(g^\prime_\infty,x^\prime_\infty)\big)\in 
G^{\hat{\bm{z}}}$, and the identity
element is $1_{G^{\hat{\bm{z}}}}=\big(1_G,(1_G,0)\big)$.
The corresponding crossed module of Lie algebras is given
by $(\g^{\hat{\bm{z}}},\g^{\hat{\bm{z}}},\id,\ad)$ with
\begin{subequations}
\begin{flalign}
\g^{\hat{\bm{z}}}\,=\, \g\times \big(\g\ltimes\g_{\mathrm{ab}}\big)\quad,
\end{flalign}
where $\g_{\mathrm{ab}}$ denotes the Abelian Lie algebra given by 
the vector space $\g$ and the trivial Lie bracket.
The Lie algebra structure reads explicitly as
\begin{flalign}
\big[\big(x_0,(x_\infty,y_\infty)\big),\big(x_0^\prime,(x_\infty^\prime,y_\infty^\prime)\big)\big]
\,=\, \big([x_0,x_0^\prime],\big([x_\infty,x_\infty^\prime], [x_\infty,y_\infty^\prime] + [y_\infty,x_\infty^\prime]\big)\big)\quad,
\end{flalign}
\end{subequations}
for all $\big(x_0,(x_\infty,y_\infty)\big),\big(x_0^\prime,(x_\infty^\prime,y_\infty^\prime)\big)\in \g^{\hat{\bm{z}}}$.

The pairing $\ipp{\,\cdot\,}{\,\cdot\,}_\omega: \g^{\hat{\bm{z}}}\otimes \g^{\hat{\bm{z}}}\to\mathbb{C}$ 
from \eqref{eqn:defectpairing higher} reads in the present example as
\begin{flalign}
\ipp{\big(x_0,(x_\infty,y_\infty)\big)}{\big(x_0^\prime,(x_\infty^\prime,y_\infty^\prime)\big)}_\omega
\,=\,\ip{x_0}{x_0^\prime} - \ip{x_\infty}{x_\infty^\prime} + \ip{x_\infty}{y_\infty^\prime} 
+ \ip{y_\infty}{x_\infty^\prime}\quad,
\end{flalign}
for all $\big(x_0,(x_\infty,y_\infty)\big),\big(x_0^\prime,(x_\infty^\prime,y_\infty^\prime)\big)\in \g^{\hat{\bm{z}}}$.
A possible choice for an isotropic crossed submodule 
$(G^\diamond,H^\diamond,\id,\Ad)\subseteq (G^{\hat{\bm{z}}},G^{\hat{\bm{z}}},\id,\Ad)$
is given by
\begin{flalign}\label{eqn:diamondEX1}
G^\diamond\,=\, G \times \big(\{1_G\}\ltimes \tilde{\g}\big)\quad,\qquad
H^\diamond\,=\, \{1_G\} \times \big(\{1_G\}\ltimes \tilde{\g}\big)\quad.
\end{flalign}
Note that this choice satisfies our maximality condition \eqref{eqn:solve conditions}.

The ansatz \eqref{eqn:admissibleansatz} for the Lax connection specializes in the present example to
\begin{subequations}\label{eqn:admissibleansatzEX1}
\begin{flalign}
A \,&=\,A_c\,=\, \sum_{i=1}^3 A_c^i\,\dd u_i\quad,\\
B \,&=\,  B_c + \frac{B_1}{z-1} \,=\, 
\sum_{i,j=1}^3\bigg(B_c^{ij} + \frac{B_1^{ij}}{z-1}\bigg)\,\dd u_i\wedge \dd u_j\quad.
\end{flalign}
\end{subequations}
Our goal is to determine the forms $A_c\in\Omega^1(M,\g)$ and $B_c,B_1\in\Omega^2(M,\g)$ 
by solving the boundary conditions \eqref{eqn:bdyconditionsj}. 
Let us consider a general edge mode field
$(k,\kappa)\in C^\infty(M,G^{\hat{\bm{z}}})\times \Omega^1(M,\g^{\hat{\bm{z}}})$
in this model, which we can write more explicitly as
\begin{flalign}
k\,=\,\big(k_0,(k_\infty,l_\infty)\big)\quad,\qquad \kappa = \big(\kappa_0,(\kappa_\infty,\lambda_\infty)\big)\quad.
\end{flalign}
This can be simplified considerably by using the gauge transformations in 
\eqref{eqn:edgemode gauge transformation}, with $\bm{\iota}$
replaced by $\bm{j}$ since we are in the context of higher-order poles,
in order to gauge fix the edge modes. For the transformation parameters 
$(g^\diamond,\gamma^\diamond)\in C^\infty(M,G^\diamond)\times \Omega^1(M,\h^\diamond)$
and $(g,\gamma)\in C^\infty(M,G)\times \Omega^1(M,\g)$ which are constant along $\CP$, 
the component $k$ transforms as
\begin{flalign}
\nonumber k^\prime\,&=\, g^\diamond\,k\,\bm{j}^\ast(g)^{-1}\,=\, \big(g_0^\diamond,(1_G,x_\infty^\diamond)\big)\,
\big(k_0,(k_\infty,l_\infty)\big)\, \big(g^{-1},(g^{-1},0)\big)\\
\, &=\,\big(g_0^\diamond\,k_0\,g^{-1},(k_\infty\,g^{-1}, x^\diamond_\infty + l_\infty)\big)\quad,
\end{flalign}
which becomes the identity $k^\prime = (1_G,(1_G,0))$ when choosing 
$g = k_\infty$, $g_0^\diamond = k_\infty\,k_0^{-1}$ and $x_\infty^\diamond = -l_\infty$.
This allows us to fix without loss of generality the gauge in which $k = 1_{G^{\hat{\bm{z}}}} 
=  (1_G,(1_G,0))$ is the identity. Under residual gauge transformations, which are characterized
by $g^\diamond = 1_{G^\diamond}$ and $g=1_G$, the component $\kappa$ transforms as
\begin{flalign}
\nonumber \kappa^\prime\,&=\, \gamma^\diamond + \kappa -\bm{j}^\ast(\gamma)
\,=\,\big(0,(0,\gamma_\infty^\diamond)\big) +  \big(\kappa_0,(\kappa_\infty,\lambda_\infty)\big)
- \big(\gamma,(\gamma,0)\big)\\
\,&=\, \big(\kappa_0-\gamma,(\kappa_\infty-\gamma, \gamma_\infty^\diamond + \lambda_\infty)\big)\quad,
\end{flalign}
which becomes $\kappa^\prime = (0,(\kappa_\infty-\kappa_0,0))$ when choosing 
$\gamma = \kappa_0$ and $\gamma^\diamond_\infty = -\lambda_\infty$.
Hence, the general form of the gauge fixed edge mode is 
\begin{flalign}
k \,=\, 1_{G^{\hat{\bm{z}}}} \,=\, (1_G,(1_G,0))\quad,\qquad \kappa \,=\, (0,(\kappa_\infty,0))\quad.
\end{flalign}
Working out the first boundary condition \eqref{eqn:bdyconditionsj} for the 
ansatz \eqref{eqn:admissibleansatzEX1} and the gauge fixed edge mode yields
\begin{subequations}
\begin{flalign}
{}^{(k,\kappa)}\bm{j}^\ast(A) \, =\,\bm{j}^\ast(A) - \kappa \,=\, \big(A_c,(A_c-\kappa_\infty,0)\big)\,\in\,
\Omega^{1}(M,\g^\diamond)\quad,
\end{flalign}
from which we deduce using also \eqref{eqn:diamondEX1} that $A_c=\kappa_\infty$. For the second boundary
condition \eqref{eqn:bdyconditionsj} we then find
\begin{flalign}
\nonumber {}^{(k,\kappa)}\bm{j}^\ast(B)\,&=\bm{j}^\ast(B) - F_M(\kappa) - 
\ad\big({}^{(k,\kappa)}\bm{j}^\ast(A),\kappa\big)\\
\,&=\,\big(B_c - B_1,\big(B_c-F_M(\kappa_\infty),B_1\big)\big)\,\in\, \Omega^{1}(M,\h^\diamond)\quad,
\end{flalign}
\end{subequations}
from which we deduce using also \eqref{eqn:diamondEX1} 
that $B_c = F_M(\kappa_\infty)$ and $B_1 = B_c = F_M(\kappa_\infty)$.
The candidate Lax connection for this model is thus given by
\begin{flalign}\label{eqn:LaxEX1}
(A,B) \,=\, \bigg(\kappa_\infty, \frac{z}{z-1}\,F_M(\kappa_\infty)\bigg)\,\in\,\overline{\Omega}^{(1,0)}_\omega(X,\g)\times \overline{\Omega}^{(2,0)}_\omega(X,\g)\quad.
\end{flalign}
Inserting this result into the action \eqref{eqn:extended action higher Mlegs} 
yields the defect action
\begin{flalign}\label{eqn:actionEX1}
S_M(\kappa_\infty) \,=\, - \int_M \ip{\kappa_\infty}{\tfrac{1}{2}\dd_M\kappa_\infty + \tfrac{1}{3!}[\kappa_\infty,\kappa_\infty]}\quad,
\end{flalign}
which coincides in this example with the usual $3$-dimensional Chern-Simons action for the edge mode 
$\kappa_\infty\in\Omega^1(M,\g)$. The equation of motion is the flatness condition
$F_M(\kappa_\infty)=0$, which implies that the Lax connection \eqref{eqn:LaxEX1}
is $M$-relative flat $\mathrm{curv}_M(A,B)=0$ when going on-shell.
We further observe that the defect action \eqref{eqn:actionEX1} is gauge-invariant
under additional gauge transformations which take the usual form 
${}^g\kappa_\infty = g\,\kappa_\infty\,g^{-1} - \dd_M g\,g^{-1}$,
for $g\in C^\infty(M,G)$. We expect that this is a remnant 
of the $2$-categorical nature of higher connections 
(see Remark \ref{rem:2gauge}), but since we currently 
do not know how to consistently include these aspects in our approach
we cannot give a precise argument or proof for this claim.

From the point of view of integrable field theory, the example 
determined by the action \eqref{eqn:actionEX1} and the corresponding 
Lax connection \eqref{eqn:LaxEX1} is somewhat trivial. When going on-shell, the Lax connection
simplifies to $(A,B)\vert_{\mathrm{on-shell}} = (\kappa_\infty,0)$, i.e.\ it does not
have any $z$-dependence. So the conserved charges, which one may construct
by taking holonomies of the Lax connection, are simply ordinary Wilson loops
for the flat Chern-Simons gauge field $\kappa_\infty$. The more involved example
which we will study in the next subsection will have more interesting integrability features.

\subsection{\label{subsec:Ward}Example: Ward equation}
In this subsection we derive and study an example of a $3$-dimensional
integrable field theory which is related to the Ward equation \cite{Wmodel1, Wmodel2}.
For this we consider the shifted tangent crossed module of Lie groups
$T[1]G := (G,\tilde{\g},1_G,\Ad)$, where we use again the notation $\tilde{\g}$ 
to distinguish between the Lie algebra $\g$ and the Abelian
Lie group $\tilde{\g}:=\g$ with group operation $+$ and identity element $0\in\g$.
The map $t = 1_G: \tilde{\g}\to G\,,~x\mapsto 1_G$ is constantly assigning the
identity element $1_G\in G$ and $\alpha = \Ad : G\times \tilde{\g} \to \tilde{\g}\,,~(g,x)\mapsto g\,x\,g^{-1}$
is the adjoint action. For the non-degenerate pairing \eqref{eqn: pairing}
we take any non-degenerate $\Ad$-invariant symmetric pairing on the Lie algebra $\g$.
For the meromorphic $1$-form $\omega\in\Omega^1(\CP)$ we choose
\begin{flalign}
\omega\,=\, \frac{z\,\prod_{i=1}^3 (z-a_i)}{(z-r)^2\,(z-s)^2}\,\dd z\,=\, 
\bigg(\frac{\ell^1_r}{(z-r)^2} + \frac{\ell^0_r}{z-r}+ \frac{\ell^1_s}{(z-s)^2}
+ \frac{\ell^0_s}{z-s}+1\bigg)\,\dd z\quad,
\end{flalign}
which has a four simple zeros at $z=0,a_1,a_2,a_3$
and three double poles at $z=r,s,\infty$,
hence Assumption \ref{assu:divisible} is satisfied.
As in the previous example of Subsection \ref{subsec:ChernSimons},
we use again a coordinate $z$ on $\CP$ in which $\infty$ is a pole of $\omega$.
(Note that the constants $\ell^1_r,\ell^0_r,\ell^1_s,\ell^0_s$
in the second expression are fixed in terms of the constants $a_1,a_2,a_3,r,s$
in the first one.)

The associated crossed module of jet groups \eqref{eqn:jetgroups}
is given by $(G^{\hat{\bm{z}}},\tilde{\g}^{\hat{\bm{z}}},1_{G^{\hat{\bm{z}}}},\Ad)$
with
\begin{subequations}
\begin{flalign}
G^{\hat{\bm{z}}}\,&=\, \big(G\ltimes \tilde{\g}\big)\times \big(G\ltimes \tilde{\g}\big)\times 
\big(G\ltimes \tilde{\g}\big)\quad,\\
\tilde{\g}^{\hat{\bm{z}}}\,&=\,\big(\tilde{\g}\times \tilde{\g}\big)\times\big(\tilde{\g}\times \tilde{\g}\big)\times \big(\tilde{\g}\times \tilde{\g}\big)\quad,
\end{flalign}
\end{subequations}
where the three factors correspond to the three double poles at $r,s,\infty$.
The group structures on the individual factors read explicitly as
\begin{flalign}
(g,x)\,(g^\prime,x^\prime)\,=\,\big(g\,g^\prime, x + g\,x^\prime\,g^{-1}\big)\quad,\qquad
(x,y)\,(x^\prime,y^\prime) \,=\, \big(x+x^\prime,y+y^\prime\big)\quad,
\end{flalign}
for all $(g,x),(g^\prime,x^\prime)
\in G\ltimes \tilde{\g}$ and $(x,y),(x^\prime,y^\prime)\in \tilde{\g}\times \tilde{\g}$, 
and the identity elements of the individual factors
are $1_{G\ltimes \tilde{\g}}=(1_G,0)$ and $1_{\tilde{\g}\times \tilde{\g}} = (0,0)$.
The action $\Ad: G^{\hat{\bm{z}}}\times \tilde{g}^{\hat{\bm{z}}}\to\tilde{g}^{\hat{\bm{z}}} $
reads on each factor as
\begin{flalign}
\Ad\big((g,x),(x^\prime,y^\prime)\big)\,=\,\big(g\,x^\prime\,g^{-1}, g\,y^\prime\,g^{-1} 
+ [x,g\,x^\prime\,g^{-1}]\big)\quad,
\end{flalign}
for all $(g,x)\in G\ltimes \tilde{\g}$ and $(x^\prime,y^\prime)\in \tilde{\g}\times \tilde{\g}$.

The corresponding crossed module of Lie algebras is given
by $(\g^{\hat{\bm{z}}},\g_{\mathrm{ab}}^{\hat{\bm{z}}},0,\ad)$ with
\begin{subequations}
\begin{flalign}
\g^{\hat{\bm{z}}}\,&=\, \big(\g\ltimes\g_{\mathrm{ab}}\big)\times \big(\g\ltimes\g_{\mathrm{ab}}\big)\times\big(\g\ltimes\g_{\mathrm{ab}}\big)\quad,\\
\g_{\mathrm{ab}}^{\hat{\bm{z}}}\,&=\, \big(\g_{\mathrm{ab}}\times \g_{\mathrm{ab}}\big)\times \big(\g_{\mathrm{ab}}\times \g_{\mathrm{ab}}\big)\times \big(\g_{\mathrm{ab}}\times \g_{\mathrm{ab}}\big)\quad.
\end{flalign}
\end{subequations}
The Lie algebra structure on $\g_{\mathrm{ab}}$ is the trivial one
and on the individual factors of $\g^{\hat{\bm{z}}}$ the Lie bracket reads as
\begin{flalign}
\big[(x,y),(x^\prime,y^\prime)\big]\,=\,\big([x,x^\prime], [x,y^\prime] + [y,x^\prime]\big)\quad,
\end{flalign}
for all $(x,y),(x^\prime,y^\prime)\in \g\ltimes\g_{\mathrm{ab}}$. The map $t^{\hat{\bm{z}}}_\ast =0$
is trivial and $\ad: \g^{\hat{\bm{z}}}\otimes\g_{\mathrm{ab}}^{\hat{\bm{z}}}\to \g_{\mathrm{ab}}^{\hat{\bm{z}}}$ 
is given on each factor by
\begin{flalign}
\ad\big((x,y),(x^\prime,y^\prime)\big) 
\,=\,\big([x,x^\prime], [x,y^\prime] + [y,x^\prime]\big)\quad,
\end{flalign}
for all $(x,y)\in \g\ltimes\g_{\mathrm{ab}}$ and $(x^\prime,y^\prime)\in \g_{\mathrm{ab}}\times \g_{\mathrm{ab}}$.

The pairing $\ipp{\,\cdot\,}{\,\cdot\,}_\omega: \g^{\hat{\bm{z}}}\otimes \g_{\mathrm{ab}}^{\hat{\bm{z}}}\to\mathbb{C}$ 
from \eqref{eqn:defectpairing higher} reads in the present example as
\begin{flalign}
\nonumber &\ipp{\big((x_r,y_r),(x_s,y_s),(x_\infty,y_\infty)\big)}{\big((x^\prime_r,y^\prime_r),(x^\prime_s,y^\prime_s),(x^\prime_\infty,y^\prime_\infty)\big)}_\omega\\
\nonumber &\qquad\qquad \,=\, \ell^0_r\,\ip{x_r}{x^\prime_r} + \ell^1_r\,\big(\ip{x_r}{y_r^\prime} + \ip{y_r}{x_r^\prime}\big)\\
\nonumber &\qquad\qquad \quad + \ell^0_s\,\ip{x_s}{x^\prime_s} + \ell^1_s\,\big(\ip{x_s}{y_s^\prime} + \ip{y_s}{x_s^\prime}\big)\\
&\qquad\qquad \quad -\big(\ell^0_r + \ell^0_s\big)\, \ip{x_\infty}{x^\prime_\infty} - 
\ip{x_\infty}{y^\prime_\infty}- \ip{y_\infty}{x^\prime_\infty}\quad,
\end{flalign}
for all $\big((x_r,y_r),(x_s,y_s),(x_\infty,y_\infty)\big)\in \g^{\hat{\bm{z}}}$
and $\big((x^\prime_r,y^\prime_r),(x^\prime_s,y^\prime_s),(x^\prime_\infty,y^\prime_\infty)\big)
\in \g_{\mathrm{ab}}^{\hat{\bm{z}}}$. A possible choice for an isotropic
crossed submodule $(G^{\diamond},H^{\diamond},1_{G^{\hat{\bm{z}}}},\Ad)\subseteq 
(G^{\hat{\bm{z}}},\tilde{\g}^{\hat{\bm{z}}},1_{G^{\hat{\bm{z}}}},\Ad)$
is given by
\begin{subequations}\label{eqn:diamondEX2}
\begin{flalign}
G^{\diamond}\,&=\,  \big(\{1_G\}\ltimes \tilde{\g}\big)\times \big(G\ltimes \tilde{\g}\big)\times 
\big(\{1_G\}\ltimes \tilde{\g}\big)\quad,\\
H^{\diamond}\,&=\, \big(\{0\}\times \tilde{\g}\big)\times \big(\{0\}\times \{0\}\big)\times 
\big(\{0\}\times \tilde{\g}\big)\quad.
\end{flalign}
\end{subequations}
Note that this choice satisfies our maximality condition \eqref{eqn:solve conditions}.

Choosing $z=0$ as the distinguished zero of $\omega$,
the ansatz \eqref{eqn:admissibleansatz} for the Lax connection specializes in the present example to
\begin{subequations}\label{eqn:admissibleansatzEX2}
\begin{flalign}
A \,&=\,\sum_{i=1}^3 \bigg(A_c^i + \frac{A^i_{a_i}}{z-a_i}\bigg)\,\dd u_i\quad,\\
B \,&=\,\sum_{i,j=1}^3\bigg(B_c^{ij} + \frac{B_{a_i}^{ij}}{z-a_i}
+ \frac{B_{a_j}^{ij}}{z-a_j} + \frac{B_0^{ij}}{z}\bigg)\,\dd u_i\wedge \dd u_j\quad.
\end{flalign}
\end{subequations}
We now determine the coefficient functions 
$A_c^i,A^i_{a_i} \in C^\infty(M,\g)$ and $B_c^{ij},B_{a_i}^{ij},B_{a_j}^{ij},B_{0}^{ij}
\in C^\infty(M,\g_{\mathrm{ab}})$ by solving the boundary conditions \eqref{eqn:bdyconditionsj}. 
For this it is again convenient to use the gauge transformations in 
\eqref{eqn:edgemode gauge transformation}, with $\bm{\iota}$
replaced by $\bm{j}$ since we are in the context of higher-order poles,
in order to gauge fix the edge modes $(k,\kappa)\in C^\infty(M,G^{\hat{\bm{z}}})\times 
\Omega^1(M,\g_{\mathrm{ab}}^{\hat{\bm{z}}})$ according to
\begin{flalign}
k \,=\,\big((k_r,0),(1_G,0),(1_G,0)\big)\quad,\qquad
\kappa\,=\, \big((\kappa_r,0),(0,\lambda_s),(\kappa_\infty,0)\big)\quad.
\end{flalign}
Working out the first boundary condition in \eqref{eqn:bdyconditionsj}
for the ansatz \eqref{eqn:admissibleansatzEX2} and the gauge fixed edge mode yields
\begin{subequations}\label{eqn:bdyEX2}
\begin{flalign}
A\vert_{z=r} \,=\, k_r^{-1}\,\dd_M k_r\quad,\qquad A\vert_{z=\infty}\,=\, 0\quad,
\end{flalign}
and the second boundary condition yields
\begin{flalign}
B\vert_{z=r}\,&=\, k_r^{-1}\,(\dd_M \kappa_r)\,k_r\quad,\qquad 
B\vert_{z=\infty}\,=\, \dd_M\kappa_\infty \quad,\qquad\\
B\vert_{z=s}\,&=\,0\quad,\qquad
\partial_z B\vert_{z=s}\,=\, \Lambda_s \,:=\,
\dd_M\lambda_s + \big[A\vert_{z=s},\lambda_s\big]\quad.
\end{flalign}
\end{subequations}
The system of equations \eqref{eqn:bdyEX2} can be solved for the coefficient functions
appearing in the ansatz \eqref{eqn:admissibleansatzEX2}, which gives
\begin{subequations}\label{eqn:bdysolutionEX2}
\begin{flalign}
A_c^i\,&=\ 0 \quad,\\[5pt]
A_{a_i}^i\,&=\ (r-a_i)\, k_r^{-1}\,\partial_{u_i} k_r\quad,\\[5pt]
B_c^{ij}\,&=\,(\dd_M\kappa_\infty)^{ij} \quad,\\[5pt]
B_{a_i}^{ij}\,&=\, \tfrac{(r-a_i)\,(s-a_i)^2}{a_i\,(a_i-a_j)}\, \Big( 
\tfrac{r\, (r-a_j)}{(r-s)^2}\, k_r^{-1}\,(\dd_M\kappa_r)^{ij}\,k_r \,-\, 
\tfrac{s\,(s-a_j)}{r-s}\,\Lambda_s^{ij} \,-\,
(\dd_M\kappa_\infty)^{ij} \Big)\quad,\\[5pt]
B_{a_j}^{ij}\,&=\, \tfrac{(r-a_j)\,(s-a_j)^2}{a_j\,(a_j-a_i)}\, \Big( 
\tfrac{r \,(r-a_i)}{(r-s)^2}\, k_r^{-1}\,(\dd_M\kappa_r)^{ij}\,k_r \,-\, 
\tfrac{s\, (s-a_i)}{r-s}\,\Lambda_s^{ij} \,-\,
(\dd_M\kappa_\infty)^{ij} \Big)\quad,\\[5pt]
B_0^{ij}\,&=\, \tfrac{r\,s^2}{a_i\,a_j}\, \Big( 
\tfrac{(r-a_i)\, (r-a_j)}{(r-s)^2}\, k_r^{-1}\,(\dd_M\kappa_r)^{ij}\,k_r \,-\, 
\tfrac{(s-a_i)\,(s-a_j)}{r-s}\,\Lambda_s^{ij} \,-\,
(\dd_M\kappa_\infty)^{ij} \Big)\quad.
\end{flalign}
\end{subequations}
Hence, we have uniquely solved the boundary conditions
for the connection $(A,B)$ from \eqref{eqn:admissibleansatzEX2} 
in terms of the edge mode $(k,\kappa)$.

Inserting this result into the action \eqref{eqn:extended action higher Mlegs},
and noting that the Chern-Simons term vanishes since  $t_\ast^{\hat{\bm{z}}} =0$ in the present example,
yields the defect action
\begin{multline}\label{eqn:actionEX2}
S_M(k_r,\kappa_r,\lambda_s,\kappa_\infty) \,=\, \int_M \bigg( \ell^1_s\,\ip{A\vert_{z=s}}{\dd_M\lambda_s + \tfrac{1}{2}\big[A\vert_{z=s},\lambda_s\big]}\\
+ \ell_r^1\,\ip{\partial_z A\vert_{z=r}}{k_r^{-1}\,(\dd_M\kappa_r)\,k_r} 
- \ip{\partial_{z^{-1}} A\vert_{z=\infty}}{\dd_M\kappa_\infty}\bigg)\quad,
\end{multline}
where the values of $A$ and its $z$ derivative at the various poles $z=r,s,\infty$
are determined in terms of the edge mode by \eqref{eqn:admissibleansatzEX2} and
\eqref{eqn:bdysolutionEX2}. The corresponding Lax connection $(A,B)$
for this theory is given by inserting \eqref{eqn:bdysolutionEX2} into \eqref{eqn:admissibleansatzEX2}.
As a consequence of Proposition \ref{prop:admissible}, the Euler-Lagrange equations
for the action \eqref{eqn:actionEX2} are equivalent to the $M$-relative flatness
conditions for the Lax connection, which in our present example read as
\begin{flalign}\label{eqn:EOMEX2}
\dd_M B + [A,B]\,=\,0\quad,\qquad \dd_M A +\tfrac{1}{2}[A,A]\,=\,0
\end{flalign}
since $t_\ast =0$.

By a slightly lengthy computation, the system of equations \eqref{eqn:EOMEX2} can be 
worked out component-wise by inserting \eqref{eqn:admissibleansatzEX2} together with 
the explicit coefficient functions given in \eqref{eqn:bdysolutionEX2}. One then finds that 
the $2$-form equation $\dd_M A +\tfrac{1}{2}[A,A]=0$ is equivalent to the flatness condition
\begin{flalign}\label{eqn:EOMEX2A}
\dd_M A\vert_{z=s} + \tfrac{1}{2}\big[A\vert_{z=s},A\vert_{z=s}\big]\,=\,0
\end{flalign}
for the connection $A\vert_{z=s} \,=\, k_r^{-1}\,\dd_M^{1,-1} k_r \in\Omega^1(M,\g)$
which is expressed here in terms of the edge mode $k_r\in C^\infty(M,G)$ and
the weighted derivative defined for general $m, n \in \ZZ$ as
$\dd_M^{m,n} := \sum_{i=1}^{3} (r-a_i)^m (s-a_i)^n \,\dd u_i \wedge\partial_{u_i}$.
The top-form equation $\dd_M B + [A,B]=0$ is equivalent to 
\begin{subequations}\label{eqn:EOMEX2B}
\begin{flalign}
(s-r)\, \ell^1_s ~ \dd_M^{0,-1} \Big(\dd_M\lambda_s + \big[A\vert_{z=s},\lambda_s\big]\Big) + P\,=\,0\quad,
\end{flalign}
where we have introduced the short-hand notation
\begin{flalign}
P\,&:=\, \dd_M^{1,0} \big( \dd_M \kappa_\infty \big) + \big[\partial_{z^{-1}} A\vert_{z=\infty}, \dd_M\kappa_\infty \big]+
\ell_r^1 ~\dd_M^{-1,0} \big( k_r^{-1}\,(\dd_M\kappa_r)\,k_r \big)\quad.
\end{flalign}
\end{subequations}

The top-form equation of motion \eqref{eqn:EOMEX2B}
is rather non-transparent and hence difficult to analyze in full generality. 
We will now show that, restricting to a special class of solutions, it is related to the Ward equation.
For this we assume that the edge mode fields $\kappa_r,\kappa_\infty\in\Omega^1(M,\g_{\mathrm{ab}})$ 
are de Rham closed, i.e.\ $\dd_M\kappa_r = \dd_M\kappa_\infty =0$, 
and that $\lambda_s = \sum_{i=1}^3\lambda_s^i \,\dd u_i$ has constant coefficient functions, 
i.e.\ $\partial_{u_i}\lambda_s^j=0$ for all $i,j\in\{1,2,3\}$.
It then follows that $P=0$, $\dd_M\lambda_s =0$ and $\dd_M^{0,-1}\lambda_s =0$, 
so the top-form equation \eqref{eqn:EOMEX2B} simplifies to
\begin{flalign}
\big[\dd_M^{0,-1}A\vert_{z=s}, \lambda_s\big]\,=\,0\quad.
\end{flalign}
Inserting $A\vert_{z=s} \,=\, k_r^{-1}\,\dd_M^{1,-1} k_r$ and working out the weighted
derivatives, one finds that this is equivalent to the equation
\begin{flalign}\label{eqn:Wardequationpre}
\,&\, \Big[(r-a_3)\, \partial_{u_2} \big(k_r^{-1}\,\partial_{u_3} k_r\big)-
(r-a_2)\, \partial_{u_3} \big(k_r^{-1}\,\partial_{u_2} k_r\big),(s-a_1)\,\lambda_s^1\Big]\\[4pt]
\nonumber +\,&\, \Big[(r-a_1)\, \partial_{u_3} \big(k_r^{-1}\,\partial_{u_1} k_r\big)-
(r-a_3)\, \partial_{u_1} \big(k_r^{-1}\,\partial_{u_3} k_r\big),(s-a_2)\,\lambda_s^2\Big]\\[4pt]
\nonumber +\,&\, \Big[(r-a_2)\, \partial_{u_1} \big(k_r^{-1}\,\partial_{u_2} k_r\big)-
(r-a_1)\, \partial_{u_2} \big(k_r^{-1}\,\partial_{u_1} k_r\big),(s-a_3)\,\lambda_s^3\Big]\,=\,0\quad.
\end{flalign}
If we now choose the constants $\lambda_s^i$ such that $\eta_s := (s-a_i)\,\lambda_s^i$, for all $i=1,2,3$,
we obtain
\begin{subequations}\label{eqn:Wardequation}
\begin{flalign}
\bigg[\sum_{i,j=1}^3 N_{ij}\, \partial_{u_i} \big(k_r^{-1}\,\partial_{u_j} k_r\big),\eta_s\bigg]\,=\,0\quad,
\end{flalign}
where $N_{ij}$ are the entries of the matrix
\begin{flalign}
N\,=\,\begin{pmatrix}
0 & r-a_2 & a_3-r\\
a_1-r & 0 & r-a_3\\
r-a_1 & a_2-r & 0
\end{pmatrix}\quad.
\end{flalign}
\end{subequations}
Decomposing $N=g + \tau$ into its symmetric part $g$ and 
antisymmetric part $\tau$ gives
\begin{flalign}\label{eqn:metric}
g \,=\,\begin{pmatrix}
0 & \tfrac{a_1-a_2}{2} & \tfrac{a_3-a_1}{2}\\
\tfrac{a_1-a_2}{2} & 0 & \tfrac{a_2-a_3}{2} \\
\tfrac{a_3-a_1}{2}& \tfrac{a_2-a_3}{2} & 0
\end{pmatrix} \quad,\qquad
\tau \,=\,\begin{pmatrix}
0 & r-\tfrac{a_1+a_2}{2} & \tfrac{a_1+a_3}{2}-r\\
\tfrac{a_1+a_2}{2}-r & 0 & r-\tfrac{a_2+a_3}{2}\\
r-\tfrac{a_1+a_3}{2} & \tfrac{a_2+a_3}{2}-r & 0
\end{pmatrix} \quad.
\end{flalign}
Considering the trace and the determinant of the symmetric part $g$,
one deduces that 1.)~the eigenvalues of $g$ are all non-zero since by hypothesis
$a_i\neq a_j$, for all $i\neq j$, and 2.)~one eigenvalue has the opposite sign
of the other two eigenvalues. This means that $g$ defines a Lorentzian metric.
Using the freedom to multiply the equation of motion \eqref{eqn:Wardequation} by $-1$,
we can assume without loss of generality that $g$ has signature $(-++)$.
We also observe that the coordinates $u_i$ which were used in describing
the admissible pole structure of the connection $(A,B)$ in \eqref{eqn:admissibleansatzEX2}
all turn out to be null with respect to the metric \eqref{eqn:metric}.
Concerning the antisymmetric part $\tau$, we apply the Hodge operator $\ast$
associated with $g$ and an arbitrary choice of orientation of $M=\mathbb{R}^3$
and find that the resulting covector $v=\ast(\tau)$ has squared norm $\vert\!\vert v\vert\!\vert^2_g =1$ 
with respect to the Lorentzian metric $g$, hence it defines a normalized spacelike covector.

In order to make the relationship between the equation of motion \eqref{eqn:Wardequation}
and Ward's equation \cite{Wmodel1, Wmodel2} more explicit, we can transform
the null coordinates $(u_1,u_2,u_3)$ to new coordinates $(t,x,y)$ in which the metric
takes the standard Minkowski form $\mathrm{diag}(-1,1,1)$. In these coordinates
the equation of motion \eqref{eqn:Wardequation} then reads as
\begin{flalign}\label{eqn:Wardequationexplicit}
&\bigg[-\partial_t\big(k_r^{-1}\,\partial_t k_r\big)
+\partial_x\big(k_r^{-1}\,\partial_x k_r\big) 
+\partial_y\big(k_r^{-1}\,\partial_y k_r\big) \\
\nonumber &\qquad \qquad +a\, \Big(\partial_x\big(k_r^{-1}\,\partial_y k_r\big) -\partial_y\big(k_r^{-1}\,\partial_x k_r\big)\Big) \\
\nonumber &\qquad \qquad +b\, \Big(\partial_y\big(k_r^{-1}\,\partial_t k_r\big) -\partial_t\big(k_r^{-1}\,\partial_y k_r\big)\Big) \\
\nonumber &\qquad\qquad  +c\, \Big(\partial_t\big(k_r^{-1}\,\partial_x k_r\big) -\partial_x\big(k_r^{-1}\,\partial_t k_r\big)\Big) 
, \eta_s\bigg]\,=\,0\quad,
\end{flalign}
where $(a,b,c)$ are the components of the normalized spacelike covector
$v$ in this choice of coordinates, i.e.\ they satisfy $-a^2 + b^2 + c^2 =1$.
The expression in the first entry of the Lie bracket in \eqref{eqn:Wardequationexplicit} 
is precisely the left-hand side of Ward's equation, including Ward's normalization 
condition for the covector $v$. Any solution
$k_r\in C^\infty(M,G)$ to Ward's equation is thus a solution
of our top-form equation of motion $\dd_M B +[A,B]=0$, provided that the other edge mode fields
$\kappa_r,\kappa_\infty$ are chosen to be de Rham closed and that $\lambda_s = \sum_{i=1}^3\lambda_s^i \,\dd u_i$ 
has constant coefficient functions such that $(s-a_i)\,\lambda_s^i=(s-a_j)\,\lambda_s^j$, 
for all $i,j\in\{1,2,3\}$. We would like to note that our equation of motion \eqref{eqn:Wardequationexplicit} 
also captures solutions to the inhomogeneous Ward equation with right-hand side
given by a current $j\in\Omega^1(M,\g)$ which lies in the kernel of $[\,\cdot\,,\eta_s]$, 
i.e.\ $[j,\eta_s]=0$.

It remains to investigate the $2$-form equation \eqref{eqn:EOMEX2A}.
Working in our original set of null coordinates $(u_1,u_2,u_3)$, the three components
of this equation can be written explicitly as
\begin{flalign}\label{eqn:Aequation}
(r-a_j)\, \partial_{u_i} \big(k_r^{-1}\,\partial_{u_j} k_r\big)-
(r-a_i)\, \partial_{u_j} \big(k_r^{-1}\,\partial_{u_i} k_r\big)\,=\,0\quad,
\end{flalign}
for $i,j\in\{1,2,3\}$  with $i<j$. 
Note that these are precisely the individual summands entering the Ward equation \eqref{eqn:Wardequationpre}.
One possible mechanism to solve both equations of motion \eqref{eqn:Aequation}
and \eqref{eqn:Wardequation}, as required for the full flatness of the Lax connection $(A,B)$,
is to consider solutions $k_r \in C^\infty(M,G)$ of the Ward equation
which are constant along one of the null coordinates $u_i$ of spacetime. 
The solutions one obtains in this way would then be `chiral' in
this chosen null direction. For such chiral solutions of the Ward equation,
one can then construct families of conserved charges by taking both $1$-dimensional
and $2$-dimensional holonomies \cite{SchreiberWaldorf,Joao,Waldorf2} 
of the associated fully flat Lax connection $(A,B)$, see also Subsection \ref{subsec:2hol}.

We would like to conclude this section by observing that, even without enforcing the 
very restrictive $2$-form equation \eqref{eqn:Aequation}, our approach leads to a family of
conserved charges for solutions to the top-form equation $\dd_M B + [A,B]=0$, and hence in particular
for general solutions to the Ward equation. The origin of these conserved charges 
lies in the fact that in the present example the Lie group $H=\tilde{\g}$ is Abelian with group operation $+$,
hence the exponential map and $2$-dimensional holonomy simplify drastically. This
allows us to build a conserved charge for every homogeneous
$\ad$-invariant polynomial $p\in (\mathrm{Sym}^n\g^\ast)^\g$ of degree $n$ on the Lie algebra $\g$:
Consider the product manifold $M^{n} = M\times \cdots \times M$ and denote by $\mathrm{pr}_i : M^{n}\to M$
the projection onto the $i$-th factor. From these data we can define the differential
form 
\begin{flalign}
p(B)\,:=\,p\big(\mathrm{pr}_1^\ast(B)\wedge\cdots \wedge\mathrm{pr}_n^\ast(B)\big)\,\in\, \Omega^{2\,n}(M^{n})\quad,
\end{flalign}
which as a consequence of the top-form equation $\dd_M B + [A,B]=0$ is closed. 
Indeed, from the Leibniz rule for $\dd_{M^{n}}$ and the Lie bracket, one observes that
\begin{flalign}
\nonumber \dd_{M^{n}}p(B) \,&=\,p\Big(\dd_{M^n}\Big(\mathrm{pr}_1^\ast(B)\wedge\cdots \wedge\mathrm{pr}_n^\ast(B)\Big)
+\Big[\sum_{i=1}^n\mathrm{pr}_i^\ast(A),\mathrm{pr}_1^\ast(B)\wedge\cdots \wedge\mathrm{pr}_n^\ast(B) \Big]\Big) \\
\,&=\,\sum_{i=1}^n p\Big(\mathrm{pr}_1^\ast(B)\wedge\cdots\wedge \mathrm{pr}_i^\ast\big(\dd_M B+[A,B]\big)\wedge\cdots \wedge\mathrm{pr}_n^\ast(B)\Big)\,=\,0\quad.
\end{flalign}
Picking any family of Cauchy surfaces $\Sigma_1,\dots,\Sigma_n\subset M$, one obtains a multi-local 
conserved charge
\begin{flalign}
Q_p(B)\,:=\,\int_{\Sigma_1\times\cdots\times\Sigma_n}p(B)
\end{flalign}
which depends meromorphically on the spectral parameter $z\in \CP$.


\addtocontents{toc}{\SkipTocEntry}
\section*{Data availability statement}
All data generated or analyzed during this study are contained in this document. 

\addtocontents{toc}{\SkipTocEntry}
\section*{Conflict of interest statement}
The authors have no conflict of interest to declare that are relevant to the content of this article.


\end{document}